\documentclass[a4paper,12pt]{article}
\usepackage{amsmath,amsthm,amsfonts,amssymb,bbm}
\usepackage{graphicx,psfrag,subfigure,color}
\usepackage{cite}
\usepackage{hyperref}

\numberwithin{equation}{section}

\newcommand{\e}{\varepsilon}
\newcommand{\Pb}{\mathbb{P}}

\newcommand{\I}{\mathrm{i}}
\newcommand{\dx}{\mathrm{d}}
\newcommand{\R}{\mathbb{R}}
\newcommand{\N}{\mathbb{N}}
\newcommand{\Z}{\mathbb{Z}}

\newcommand{\Id}{\mathbbm{1}}
\newcommand{\Ai}{\mathrm{Ai}}
\renewcommand{\Re}{\operatorname{Re}}

\newcommand{\Or}{\mathcal{O}}

\newtheorem{assumption}{Assumption}
\newtheorem{prop}{Proposition}[section]
\newtheorem{thm}[prop]{Theorem}
\newtheorem{lem}[prop]{Lemma}

\newtheorem{cor}[prop]{Corollary}

\newtheorem{cla}[prop]{Claim}

\newtheorem{rem}[prop]{Remark}
\newenvironment{remark}{\begin{rem}\normalfont}{\end{rem}}

\title{Anomalous shock fluctuations in TASEP\\ and last passage percolation models}
\author{Patrik L.\ Ferrari\thanks{Institute for Applied Mathematics, Bonn University, Endenicher Allee 60, 53115 Bonn, Germany. E-mail: {\tt ferrari@uni-bonn.de}} \and
Peter Nejjar\thanks{Institute for Applied Mathematics, Bonn University, Endenicher Allee 60, 53115 Bonn, Germany. E-mail: {\tt nejjar@uni-bonn.de}}
}

\date{}

\begin{document}
\maketitle
\sloppy

\vfill
\begin{abstract}
We consider the totally asymmetric simple exclusion process with initial conditions and/or jump rates such that shocks are generated. If the initial condition is deterministic, then the shock at time $t$ will have a width of order $t^{1/3}$. We determine the law of particle positions in the large time limit around the shock in a few models. In particular, we cover the case where at both sides of the shock the process of the particle positions is asymptotically described by the Airy$_1$ process. The limiting distribution is a product of two distribution functions, which is a consequence of the fact that at the shock two characteristics merge and of the slow decorrelation along the characteristics. We show that the result generalizes to generic last passage percolation models.
\end{abstract}
\vfill

\newpage

\section{Introduction}

We start by considering the simplest non-reversible interacting stochastic particle system, namely  the totally asymmetric simple exclusion process (TASEP) on $\Z$. Despite its simplicity, this model is full of interesting features. In TASEP, particles independently try to jump to their right neighbor site at a constant rate and jumps occur if the exclusion constraint is satisfied: no site can be occupied by more than one particle. Under hydrodynamic scaling, the particle density solves the deterministic Burgers equation (see e.g.~\cite{Li99,AV87}). This model belongs to the Kardar-Parisi-Zhang (KPZ) universality class~\cite{KPZ86} (see~\cite{cor11} for a recent review).

We are interested in the fluctuations around the macroscopic behavior given in terms of the solution of the Burgers equation and we focus on the fluctuations of particles' positions. Depending on the initial condition, the deterministic solution may have parts of constant and decreasing density, as well as a discontinuity, also referred to as shock. The fluctuations of the shock location have attracted a lot of attention.

For TASEP product Bernoulli measures are the only translation invariant stationary measures~\cite{Lig76}. In the first works one considered initial configurations to have a shock at the origin, with Bernoulli measures with density $\rho$ (resp.\ $\lambda$) at its left (resp.\ right), with $\rho<\lambda$.
The shock location is often identified by the position of a second class particle. In this case, the shock fluctuations are Gaussian in the scale $t^{1/2}$~\cite{Fer90,FF94b,PG90}. Microscopic information on the shock are available too~\cite{DJLS93,FKS91,Fer86,BS13}. The origin of the $t^{1/2}$ fluctuations lies in the randomness of the initial conditions, since fluctuations coming from the dynamics grow only as $t^{1/3}$. If the initial randomness is only at one side of the shock, a similar picture still holds. For example, in~\cite{BFS09} one considers the initial condition is Bernoulli-$\rho$ to the right and periodic with density $1/2$ to the left of the origin. When $\rho>1/2$ there is a shock with Gaussian fluctuations in the scale $t^{1/2}$. In that work, the fluctuations of the shock position are derived from the ones of the particle positions. The result fits in with the heuristic argument in~\cite{Spo91} (Section~5). The Gaussian form of the distribution function is not robust (see for instance Remark~17 in~\cite{BFS09}).

This paper is the first where the fluctuation laws around a shock occurring without initial randomness are analyzed. In that case, one heuristically expects that the shock fluctuations, but also tagged particles fluctuations, live only on a scale of order $t^{1/3}$, see~\cite{vanB91} for a physical argument. We find that the distribution function of a particle position (and also of tagged particles) is a product of two other distribution functions.
The reason of the product form of the distribution function is that (1) at the shock two characteristics merge and (2) along the characteristics decorrelation is slow~\cite{Fer08,CFP10b}.

More precisely, if we look at the history of a particle close to the shock at time $t$, it has non-trivial correlations with a region of width $\Or(t^{2/3})$ around the characteristics, see Figure~\ref{FigShockTASEP}.
\begin{figure}
\begin{center}
\psfrag{x}[cc]{$x$}
\psfrag{t}[lc]{$t$}
\psfrag{E}[cc]{$E$}
\psfrag{El}[cc]{$E_\ell$}
\psfrag{Er}[cc]{$E_r$}
\includegraphics[height=5cm]{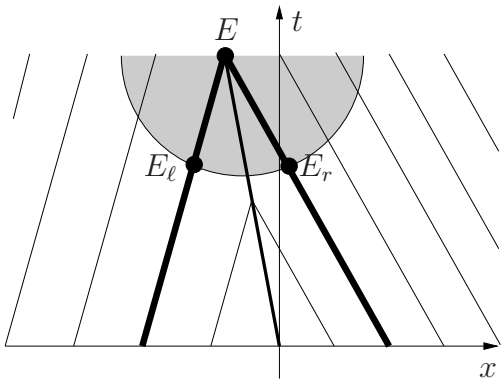}
\caption{Illustration of the characteristics for TASEP. $E$ is the shock location, where two characteristics merges (the thick lines). The gray region is of order $t^\nu$ for some $2/3<\nu<1$. Due to the slow decorrelation along characteritics, at large time $t$ the fluctuations at $E$ originates from the ones at $E_\ell$ and $E_r$.}
\label{FigShockTASEP}
\end{center}
\end{figure}
At the shock the two characteristics come together with a positive angle so that at time $t-t^{\nu}$, \mbox{$2/3<\nu<1$}, their distance will be farther away than  $\Or(t^{2/3})$. This implies that the fluctuations built up along the two characteristics before time $t-t^\nu$ will be (asymptotically) independent. But if we stay on a characteristic, then the dynamical fluctuations created between time $t-t^\nu$ and time $t$ are only $o(t^{1/3})$, which are irrelevant with respect to the total fluctuations present at time $t-t^\nu$ that are of order $t^{1/3}$ (this is also known as the slow-decorrelation phenomenon~\cite{Fer08,CFP10b}).

To generate a shock between two regions of constant density, we consider the initial condition  where $2\Z$ is fully occupied and where the jump rates of particles starting to the left (resp.\ right) of the origin is equal to $1$ (resp.\ $\alpha<1$).

In Corollary~\ref{Cor1b} we determine the distribution function of the fluctuations of TASEP particles in the $t^{1/3}$ scale. It is a product of two GOE Tracy-Widom distributions, $F_1$, and the transition from a single $F_1$ to the second $F_1$ distribution occurs over a distance of order $t^{1/3}$. In Corollary~\ref{Cor2b} (resp.~\ref{Cor3b}) we study another shock situation, where the distribution function goes from $F_2$ to $F_1$ (resp.\ from $F_2$ to $F_2$).

For completeness, let us briefly discuss what happens when there is no shock. From the KPZ theory, the fluctuations of particle positions generated by the dynamics grow as $t^{1/3}$ where $t$ is the time. When the density is constant and the initial condition is non-random, the fluctuations will be governed by $F_1$, as shown in a few cases in~\cite{BR99} (see~\cite{Sas05,BFPS06,BFP06} for joint distributions). However, when the density is decreasing, the fluctuations are governed by $F_2$, the GUE Tracy-Widom distribution, as shown in~\cite{Jo00b} (see also~\cite{BC09,PS01} for random initial conditions and~\cite{BF07,BFP06,Jo03b} for joint distributions). The transition process between these two cases has been determined in~\cite{BFS07}. The stationary initial condition~\cite{Lig76} has also constant density but with different fluctuation laws, see~\cite{BR00,FS05a,BFP09}. For further details see e.g.\ the review~\cite{Fer07}.

It is well known that TASEP is linked to last passage percolation (LPP). As a consequence slow-decorrelation holds also for generic LPP models~\cite{CFP10b}. For this reason our first main result, Theorem~\ref{ThmMain}, is a generic statement proven under some assumptions, that one needs to verify model by model. This theorem states that the distribution function of a generic last passage percolation is the product of two distribution functions corresponding to two simpler last passage problems. The verification of the assumptions can be quite involved. Here we prove they are valid for three different LPP models and obtain Corollaries~\ref{Cor1}--\ref{Cor3}.

 Using the connection to TASEP we then restate the results in terms of TASEP particles (see Corollaries~\ref{Cor1b}--\ref{Cor3b}). Corollary~\ref{Cor1b} is the result which motivated our work.

\subsubsection*{Outline}
In Section~\ref{sectResults} we define precisely the models and state our main results. In Section~\ref{sectProofGeneric} we prove the generic theorem. In order to apply it to specific models we have to verify the assumptions, which is the content of Section~\ref{SectAllLPP}. Finally, in Section~\ref{sectDistrLminus} we derive a kernel used in Section~\ref{SectAllLPP}.

\subsubsection*{Acknowledgments}
The authors would like to thank Alexei Borodin and Tomohiro Sasamoto for early discussions about the problem, and Herbert Spohn for valuable remarks. The work of P.L.~Ferrari was supported by the German Research Foundation via the SFB 1060--B04 project. P.~Nejjar is supported by the Bonn International Graduate School (BIGS).

\section{Models and main results}\label{sectResults}

\subsection{Last passage percolation -- general statement}
We consider last passage percolation (LPP) models\footnote{One can also define LPP models on $\R^2$, see e.g.~\cite{CFP10b}. The general arguments in this paper are unchanged. Only a minor modification in the proofs in Section~\ref{sectProofNocrossing} is needed, namely the discretization used in Johansson's argument~\cite{Jo00}.} on $\Z^2$ with independent random variables $\{\omega_{i,j},i,j\in\Z\}$. An up-right path $\pi=(\pi(0),\pi(1),\ldots,\pi(n))$ on $\Z^2$ from a point $A$ to a point $E$ is a sequence of points in $\Z^2$ with \mbox{$\pi(k+1)-\pi(k)\in \{(0,1),(1,0)\}$}, with $\pi(0)=A$ and $\pi(n)=E$, and where $n$ is called the length $\ell(\pi)$ of $\pi$. Now, given two sets of points $S_A$ and $S_E$, one defines the last passage time $L_{S_A\to S_E}$ as
\begin{equation}\label{eq1}
L_{S_A\to S_E}=\max_{\begin{subarray}{c}\pi:A\to E\\A\in S_A,E\in S_E\end{subarray}} \sum_{1\leq k\leq \ell(\pi)} \omega_{\pi(k)}.
\end{equation}
The purpose of this paper is to determine the law of last passage times of various models.
Finally, we denote by $\pi^{\rm max}_{S_A\to S_E}$ any maximizer of the last passage time $L_{S_A\to S_E}$. For continuous random variables, the maximizer is a.s.\ unique.

In this paper we consider situations where the end set is one point and the starting set is the union of sets, namely\footnote{We will not write always explicitly the integer parts.}
\begin{equation}
S_A=\mathcal{L}^+\cup \mathcal{L}^-,\quad S_E=E=(\lfloor\eta  t\rfloor,\lfloor  t\rfloor),
\end{equation}
where $\mathcal{L}^+ \subseteq \{(v,n)\in \mathbb{Z}^{2}:v\leq 0, n \geq  0\}$, $\mathcal{L}^- \subseteq \{(v,n)\in \mathbb{Z}^{2}:n\leq 0, v\geq 0\}$.
Note that, by putting some of the $\omega_{i,j}$ to zero, it is always possible to choose $\mathcal{L}^+=(\Z_-,0)$ and $\mathcal{L}^-=(0,\Z_-)$.

With this choice it follows from the definition of the last passage time (\ref{eq1}) that
\begin{equation}
L=L_{S_A\to S_E} = \max\left\{L_{\mathcal{L}^+\to (\eta t, t)},L_{\mathcal{L}^-\to (\eta t, t)}\right\}.
\end{equation}
The two random variables $L_1=L_{\mathcal{L}^+\to (\eta t, t)}$ and $L_2=L_{\mathcal{L}^-\to (\eta t, t)}$ are not independent. However, under some assumptions they are essentially independent as $t\to\infty$, in the sense that the random last passage time $L=\max\{L_1,L_2\}$ properly rescaled has asymptotically the law of the product of the two rescaled random variables. This is due to the fact that the fluctuations present in the region where the maximizers of the two LPP problems tend to come together are on a smaller scale than the typical fluctuations. This is by  virtue of the slow-decorrelation phenomenon~\cite{Fer08,CFP10b}.

Typically one has in mind a law of large numbers  $L_i/ t \to \mu_i$ as \mbox{$t\to\infty$} and a fluctuation result $L_i-\mu_i  t=\Or( t^{\chi_i})$ with $\chi_i=1/3$ or $\chi_i=1/2$. If $L_1$ and $L_2$ have different leading orders $\mu_1,\mu_2$, the result is quite easy since only the largest of the two random variables is relevant in the $ t\to\infty$ limit. This situation can be treated directly with coupling arguments as in~\cite{BC09}. If $\mu_1=\mu_2=\mu$ but for instance $\chi_1<\chi_2$, then the natural scaling is \mbox{$(L-\mu  t)/ t^{\chi_2}$}, under which scaling $(L_1-\mu  t)/ t^{\chi_2}$ degenerates to the trivial random variable $0$ and acts as a cut-off. This situation occured for instance in~\cite{BFS09} (Proposition~1).

In this paper we consider the case where $L_1$ and $L_2$ have the same leading order $\mu$ and both fluctuations live  in the scale $ t^{1/3}$. This is our first assumption.

\begin{assumption}\label{Assumpt1}
Assume that there exists some $\mu$ such that
\begin{equation}\label{eq4a}
\lim_{ t\to\infty} \Pb\left(\frac{L_{\mathcal{L}^+\to (\eta t, t)}-\mu  t}{ t^{1/3}}\leq s\right) = G_1(s),
\end{equation}
and
\begin{equation}\label{eq4b}
\lim_{ t\to\infty} \Pb\left(\frac{L_{\mathcal{L}^-\to (\eta t, t)}-\mu  t}{ t^{1/3}}\leq s\right) = G_2(s),
\end{equation}
where $G_1$ and $G_2$ are some distribution functions.
\end{assumption}

Secondly, we assume that there is a point $E^+$ at distance of order $ t^\nu$, for some $1/3<\nu<1$, which lies on the characteristic  from $\mathcal{L}^+$ to $E$ and that there is slow-decorrelation as in Theorem~2.1 of~\cite{CFP10b}.
\begin{assumption}\label{Assumpt2}
Assume that we have a point $E^+=(\eta  t-\kappa  t^\nu, t- t^\nu)$ such that for some $\mu_0$, and $\nu\in (1/3,1)$ it holds
\begin{equation}\label{eq5}
\begin{aligned}
\lim_{ t\to\infty} \Pb\left(\frac{L_{E^+\to (\eta t, t)}-\mu_0  t^\nu}{ t^{\nu/3}}\leq s\right) &= G_0(s),\\
\lim_{ t\to\infty} \Pb\left(\frac{L_{\mathcal{L}^+\to E^+}-\mu  t+\mu_0  t^\nu}{ t^{1/3}}\leq s\right) &= G_1(s),
\end{aligned}
\end{equation}
where $G_0$ and $G_1$ are distribution functions.
\end{assumption}

Then, provided (\ref{eq4a}) and (\ref{eq5}) hold, Theorem~2.1 of~\cite{CFP10b} implies that for any $M>0$,
\begin{equation}
\lim_{ t\to\infty}\Pb\left(|L_{\mathcal{L}^+\to (\eta t, t)}-L_{\mathcal{L}^+\to E^+}-\mu_0  t^\nu|\geq M  t^{1/3}\right)=0.
\end{equation}
This means that the fluctuations of $L_{\mathcal{L}^+\to (\eta t, t)}$ are the same as the ones of $L_{\mathcal{L}^+\to E^+}$ up to $o( t^{1/3})$. Thus, we have to determine the maximum of $L_{\mathcal{L}^+\to E^+}$ and $L_{\mathcal{L}^-\to E}$. The final assumption ensures that these two random variables are asymptotically independent.
\begin{assumption}\label{Assumpt3}
Let $\nu$ be as in Assumption~\ref{Assumpt2}. Consider the points \mbox{$D_{\gamma}=(\lfloor \gamma \eta  t \rfloor,\lfloor  \gamma  t\rfloor)$} with \mbox{$\gamma \in [0,1- t^{\beta-1}]$}. Assume that there exists a \mbox{$\beta \in (0,\nu)$}, such that
\begin{equation}
\begin{aligned}
\lim_{ t\to\infty}\Pb\bigg(\bigcup_{D_{\gamma}\atop \gamma \in [0,1- t^{\beta-1}]}\left\{D_\gamma\in \pi^{\rm max}_{L_{\mathcal{L}^+\to E^+}}\right\}\bigg) &=0,\\
\lim_{ t\to\infty}\Pb\bigg(\bigcup_{D_{\gamma}\atop \gamma \in [0,1- t^{\beta-1}]}\left\{D_\gamma\in \pi^{\rm max}_{L_{\mathcal{L}^-\to (\eta t, t)}}\right\}\bigg) &=0.
\end{aligned}
\end{equation}
\end{assumption}

Under these assumptions, that will be verified in special cases, we have the first result of this paper, proven in Section~\ref{sectProofGeneric}.
\begin{thm}\label{ThmMain}
Under Assumptions~\ref{Assumpt1}--\ref{Assumpt3} we have
\begin{equation}
\lim_{ t\to\infty} \Pb\left(\frac{\max\left\{L_{\mathcal{L}^+\to (\eta t, t)},L_{\mathcal{L}^-\to (\eta t, t)}\right\}-\mu  t}{ t^{1/3}}\leq s\right) = G_1(s) G_2(s).
\end{equation}
\end{thm}

\subsection{Application to specific LPP models}\label{sectApplicLPP}
Let us consider now $\omega_{i,j}$ to be exponentially distributed random variables, that will become waiting times for TASEP particles. Let the waiting times be given by
\begin{equation}\label{eq10}
\begin{aligned}
\omega_{i,j}\sim \exp(1),\quad & j\geq 1,\\
\omega_{i,j}\sim \exp(\alpha),\quad & j\leq 0,
\end{aligned}
\end{equation}
for some $\alpha>0$. We are going to consider the scaling
\begin{equation}\label{eq11}
\eta=\eta_0+u t^{-2/3}.
\end{equation}
Then the following results hold true and will be proven in Section~\ref{sectApplications}, see Figure~\ref{FigLPPpictures} for an illustration of the geometry in the following three corollaries.
\begin{figure}
\begin{center}
\psfrag{0}[c][b]{$0$}
\psfrag{P1}[c][l]{$(\frac{\alpha}{2-\alpha}t,t)$}
\psfrag{P2}[c][l]{$(\frac{\alpha(3-2\alpha)}{2-\alpha}t,t)$}
\psfrag{P3}[c][l]{$(t,t)$}
\psfrag{a}[c][b]{$(a)$}
\psfrag{b}[c][b]{$(b)$}
\psfrag{c}[c][b]{$(c)$}
\includegraphics[height=5cm]{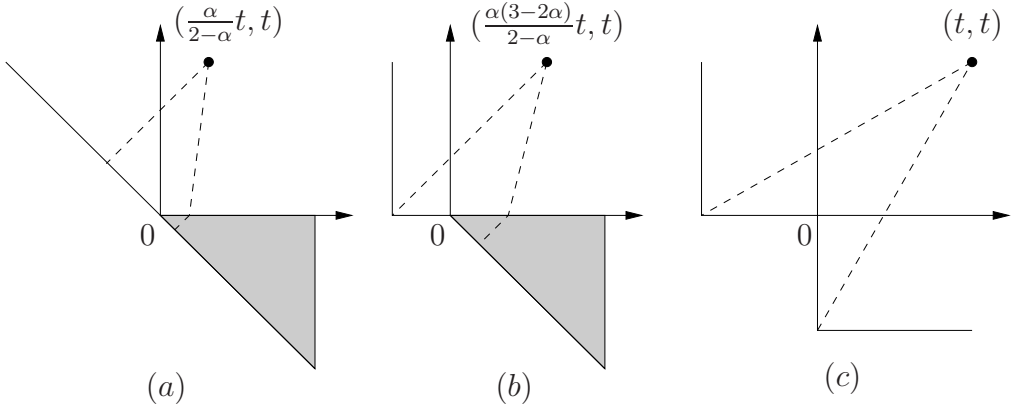}
\caption{Illustration of the geometry considered in (a) Corollary~\ref{Cor1}, (b) Corollary~\ref{Cor2}, and (c) Corollary~\ref{Cor3}, for $u=b=0$ and $\alpha=1/2$. The random variables in the gray (resp.\ white) regions are $\exp(\alpha)$ (resp.\ $\exp(1)$) distributed. The dashed lines represents the typical trajectories of the maximizers for the two LPP problems.}
\label{FigLPPpictures}
\end{center}
\end{figure}

\begin{cor}[Two point-to-line problems]\label{Cor1}
Let
\begin{equation}
\mathcal{L}^+=\{(-v,v),v\in\Z_+\},\quad \mathcal{L}^-=\{(-v,v),v\in\Z_-\},
\end{equation}
with $\eta_0=\frac{\alpha}{2-\alpha}$ and $\alpha<1$. Then, Theorem~\ref{ThmMain} holds with $\mu=4/(2-\alpha)$ and
\begin{equation}
G_1(s)=F_1\left(\frac{s-2u}{\sigma_1}\right),\quad G_2(s)=F_1\left(\frac{s-2u/\alpha}{\sigma_2}\right),
\end{equation}
where $F_1$ is the GOE Tracy-Widom distribution function of random matrices~\cite{TW96}, $\sigma_1=\frac{2^{2/3}}{(2-\alpha)^{1/3}}$ and $\sigma_2=\frac{2^{2/3}(2-2\alpha+\alpha^2)^{1/3}}{\alpha^{2/3}(2-\alpha)}$.
\end{cor}

\begin{cor}[One point-to-point and one point-to-line problem]\label{Cor2}
Let
\begin{equation}
\mathcal{L}^+=([-\lfloor \beta t\rfloor,0],0)\cup(-\lfloor \beta t\rfloor,\Z_+),\quad \mathcal{L}^-=\{(-v,v),v\in\Z_-\},
\end{equation}
with $\beta=\beta_0+b t^{-2/3}$, $\beta_0=1-\eta_0$, $\eta_0=\frac{\alpha (3-2 \alpha)}{2-\alpha}$ and $\alpha\in (0,1)$. Then, Theorem~\ref{ThmMain} holds with $\mu=4$ and
\begin{equation}
G_1(s)=F_2\left(\frac{s-2(u+b)}{\sigma_1}\right),\quad G_2(s)=F_1\left(\frac{s-2u/\alpha}{\sigma_2}\right),
\end{equation}
where $F_2$ is the GUE Tracy-Widom distribution function of random matrices~\cite{TW94}, $\sigma_1=2^{4/3}$, and $\sigma_2=\frac{2^{2/3}(6-10\alpha+6\alpha^2-\alpha^3)^{1/3}}{\alpha^{2/3}(2-\alpha)}$.
\end{cor}

\begin{cor}[Two point-to-point problems]\label{Cor3}
Let us fix a $\beta>0$ and consider
\begin{equation}
\mathcal{L}^+=(-\lfloor \beta t\rfloor ,\Z_+)\cup([-\lfloor \beta t\rfloor,0],0),\quad \mathcal{L}^-=(0,[0,-\lfloor \beta t\rfloor])\cup(\Z_+,-\lfloor \beta t\rfloor),
\end{equation}
with $\eta_0=1$ and $\alpha=1$. Then, Theorem~\ref{ThmMain} holds with $\mu = (1+\sqrt{1+\beta})^2$ and
\begin{equation}
G_1(s)=F_2\left(\frac{s-u(1+1/\sqrt{1+\beta})}{\sigma}\right),\quad G_2(s)=F_2\left(\frac{s-u(1+\sqrt{1+\beta})}{\sigma}\right),
\end{equation}
where $\sigma=(1+\sqrt{1+\beta})^{4/3}/(1+\beta)^{1/6}$.
\end{cor}

\subsection{Application to the totally asymmetric simple exclusion process}
It is well known that the choice of $\omega_{i,j}$ to be exponential distributed random variables directly links LPP with the totally asymmetric simple exclusion process (TASEP), which we recall here. TASEP is an interacting particle system on $\Z$ where two particles can not occupy the same site at the same time. Further, particles (independently) try to jump to their right neighboring site after an exponentially distributed waiting time. A jump occurs only if the destination site is empty. As a consequence, the order of particles is preserved. Thus, we can assign to each particle a number and we do it from right to left, i.e.
\[\ldots < x_2(0) < x_1(0) < 0 \leq x_0(0)< x_{-1}(0)< \cdots.\] Then, at all time $t\geq 0$, $x_{n+1}(t)<x_n(t)$, $n\in\Z$.  Then, the precise link between LPP and TASEP is the following. Let $\omega_{i,j}$ be the exponential  waiting time of particle $j$, and $L_E$ be the last passage time from $S_A=\{(u,k)\in\Z^2: u=k+x_k(0), k\in\Z\}$. This implies that
\begin{equation}\label{eqLinkLPPtasep}
\Pb\left(L_{{S_A}\to (m,n)}\leq t\right)=\Pb\left(x_n(t)+n\geq m\right).
\end{equation}
This connection will be used several times to verify that Assumptions~\ref{Assumpt1}--\ref{Assumpt3} hold in special cases.

The particular choice of the $\omega_{i,j}$ in (\ref{eq10}) means that particles with label $n\geq 1$ have jump rate $1$, while particles with label $n\leq 0$ have jump rate $\alpha$. The choice (\ref{eq11}) implies that we look at particle number $t$ at different times. Indeed, if
\begin{equation}
\lim_{t\to\infty}\Pb\left(L_{S_A\to (\eta_0 t+ut^{1/3},t)}\leq \mu t + s t^{1/3}\right)=F(u,s),
\end{equation}
then by (\ref{eqLinkLPPtasep}) we have that
\begin{equation}\label{eqFullLink}
\lim_{t\to\infty}\Pb\left(x_t(\mu t+\tau t^{1/3})\geq (\eta_0-1)t - s t^{1/3}\right)=F(-s,\tau).
\end{equation}
Since this relation is straightforward we do not restate the three corollaries for the tagged particle problem. Instead, we restate them so that they gives the distribution function at a fixed time $t$ of particles around the shock.

In the case of Corollary~\ref{Cor2}-\ref{Cor3}, the boundaries of the LPP problem to $(\eta t,t)$ also depends on the variable $t$. This has to be taken in account here too. Therefore, let us write explicitly this dependence in the measure and just write $L_{m,n}$ for the last passage time. For the case of Corollary~\ref{Cor1}, the boundary condition does not depends on the observation time parameter $t$. For this case, one can just set $\beta=0$ in the computations below. Assume that we have, as in the previous section,
\begin{equation}
\lim_{t\to\infty}\Pb_{\beta t}(L_{\eta_0 t+u t^{1/3},t}\leq \mu(\beta) t+s t^{1/3})= F(\beta,u,s).
\end{equation}
By (\ref{eqLinkLPPtasep}) we have
\begin{equation}\label{eq2.21}
\Pb_{\beta t}(x_{\nu t+\xi t^{1/3}}(t)\geq v t - s t^{1/3})=\Pb_{\beta t}(L_{(\nu+v)t+(\xi-s)t^{1/3},\nu t + \xi t^{1/3}}\leq t)
\end{equation}
Let us define $\tilde t$, $\eta$, and $\tilde\beta$ by the equations
\begin{equation}
\tilde t=\nu t+\xi t^{1/3},\quad
\eta\tilde t= (\nu+v)t+(\xi-s) t^{1/3},\quad
\tilde\beta\tilde t =\beta t.
\end{equation}
This gives, $t=\tilde t/\nu-\xi\nu^{-4/3}\tilde t^{1/3}+\Or(\tilde t^{-1/3})$, from which
\begin{equation}
\begin{aligned}
\eta&=(1+v/\nu)-(s+\xi v/\nu) \nu^{-1/3} \tilde t^{-2/3},\\
\tilde\beta&=\beta/\nu-\xi \beta \nu^{-4/3} \tilde t^{-2/3},
\end{aligned}
\end{equation}
up to $\Or(\tilde t^{-4/3})$. By plugging this in (\ref{eq2.21}) one readily obtains
\begin{equation}
\begin{aligned}
\lim_{t\to\infty}\Pb_{\beta t}(x_{\nu t+\xi t^{1/3}}(t)\geq v t - s t^{1/3}) &=
\lim_{\tilde t\to\infty}\Pb_{\tilde \beta\tilde t}(L_{\eta_0\tilde t+u \tilde t^{1/3},\tilde t}\leq \tilde t/\nu-\xi\nu^{-4/3}\tilde t^{1/3})\\
&=\lim_{\tilde t\to\infty}\Pb_{\tilde \beta\tilde t}(L_{\eta_0\tilde t+u \tilde t^{1/3},\tilde t}\leq \mu(\tilde\beta)\tilde t+\tilde s\tilde t^{1/3})
\end{aligned}
\end{equation}
with
\begin{equation}
\eta_0=1+v/\nu,\quad u=-(s+\xi v/\nu) \nu^{-1/3}, \quad \tilde s=\xi(\beta\mu'(\beta/\nu)-1)\nu^{-4/3}.
\end{equation}
provided that it holds $\mu(\beta/\nu)=1/\nu$. This condition sets which particles are around the shock position at time $t$. Then, by (\ref{eq2.21}) we have
\begin{equation}\label{eqFullLinkSpatial}
\lim_{t\to\infty}\Pb_{\beta t}(x_{\nu t+\xi t^{1/3}}(t)\geq v t - s t^{1/3}) = F\left(\frac{\beta}{\nu},-\frac{s+\xi v/\nu}{\nu^{1/3}},\frac{\xi(\beta\mu'(\beta/\nu)-1)}{\nu^{4/3}}\right).
\end{equation}
\begin{figure}
\begin{center}
\psfrag{0}[c][b]{$0$}
\psfrag{density}[l][l]{Density $\rho$}
\psfrag{position}[c][l]{Position}
\psfrag{v1}[c][c]{$-\tfrac{1-\alpha}{2}t$}
\psfrag{v2}[c][c]{$-\tfrac{(1-\alpha)^2}{2(2-\alpha)}t$}
\psfrag{b1}[c][c]{$-\beta t$}
\psfrag{b2}[c][c]{$\beta t$}
\psfrag{alpha}[c][c]{$\tfrac{\alpha}{2} t$}
\psfrag{rhol1}[c][c]{$\rho=\tfrac12$}
\psfrag{rhor1}[c][c]{$\rho=\tfrac{2-\alpha}{2}$}
\psfrag{rhol2}[c][c]{$\rho=\tfrac12$}
\psfrag{r1}[l][c]{$\rho=1$}
\psfrag{rhor2}[c][c]{$\rho=\tfrac{2-\alpha}{2}$}
\psfrag{rhol3}[l][c]{$\rho=\tfrac{1-\beta}{2}$}
\psfrag{rhor3}[c][c]{$\rho=\tfrac{1+\beta}{2}$}
\psfrag{a}[c][b]{$(a)$}
\psfrag{b}[c][b]{$(b)$}
\psfrag{c}[c][b]{$(c)$}
\includegraphics[height=6cm]{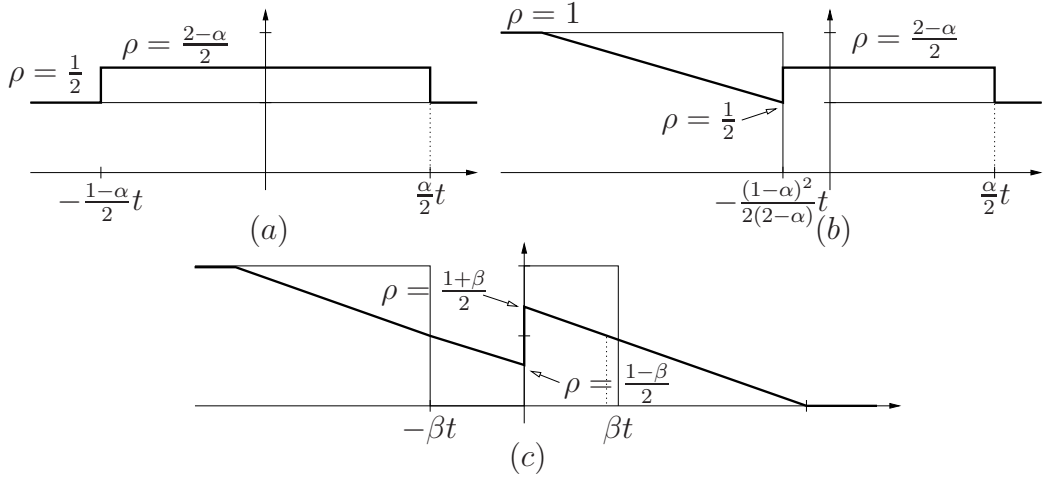}
\caption{The thick lines are the density profiles $\rho$ at time $t$ (resp. $t=\ell$) of (a) Corollary~\ref{Cor1b}, (b) Corollary~\ref{Cor2b}, and (c) Corollary~\ref{Cor3b}, for $u=b=0$ and $\alpha=1/2$. The thin lines are the initial conditions. The dotted vertical lines indicate the macroscopic position of the particle that started from the origin.
}
\label{FigTASEPpictures}
\end{center}
\end{figure}

The LPP situations considered above correspond to cases where, in terms of TASEP, there is a macroscopic discontinuity in the particles' density, i.e., there is a shock. Using (\ref{eqFullLinkSpatial}) we can restate Corollaries~\ref{Cor1}--\ref{Cor3} in terms of TASEP as follows, see Figure~\ref{FigTASEPpictures} for an illustration of the density profiles.

\medskip

\begin{cor}[At the $F_1$--$F_1$ shock]\label{Cor1b} Let $x_n(0)=-2n$ for $n\in\Z$. For $\alpha<1$ let $\nu=\frac{2-\alpha}{4}$ and $v=-\frac{1-\alpha}{2}$. Then it holds
\begin{equation}\label{eqCor1b}
\lim_{t\to\infty}\Pb\left(x_{\nu t+\xi t^{1/3}}(t)\geq v t - s t^{1/3}\right)=F_1\left(\frac{s-\xi/\rho_1}{\sigma_1}\right)F_1\left(\frac{s-\xi/\rho_2}{\sigma_2}\right),
\end{equation}
with $\rho_1=\frac12$, $\rho_2=\frac{2-\alpha}{2}$, $\sigma_1=\frac{1}{2}$, and $\sigma_2=\frac{\alpha^{1/3}(2-2\alpha+\alpha^2)^{1/3}}{2(2-\alpha)^{2/3}}$.
\end{cor}
As one can see from (\ref{eqCor1b}) the shock moves with speed $v$. When $\xi$ is very large we are in the region before the shock, where the density of particle is $1/2$. Indeed, by replacing $s\to s+2\xi$ and taking the $\xi\to \infty$ limit, then (\ref{eqCor1b}) converges to $F_1(s/\sigma_1)$. Similarly, when $-\xi$ is very large we are already in the shock, where the density of particles in $(2-\alpha)/2$. Indeed, by replacing $s\to s+2\xi/(2-\alpha)$ and taking $\xi\to-\infty$, then (\ref{eqCor1b}) converges to  $F_1(s/\sigma_2)$. This is the reason why we call this situation a $F_1$--$F_1$ shock.

\medskip

\begin{cor}[At the $F_2$--$F_1$ shock]\label{Cor2b}
For $\alpha<1$ let $\nu=1/4$ and \mbox{$v=-\frac{(1-\alpha)^2}{2(2-\alpha)}$}. Let $x_n(0)=v \ell -n$ for $n\geq 1$ and \mbox{$x_n(0)=-2n$} for $n\leq 0$.  Then it holds
\begin{equation}
\lim_{\ell\to\infty}\Pb\left(x_{\nu \ell+\xi \ell^{1/3}}(t=\ell)\geq v \ell - s \ell^{1/3}\right)=F_2\left(\frac{s-\xi/\rho_1}{\sigma_1}\right) F_1\left(\frac{s-\xi/\rho_2}{\sigma_2}\right),
\end{equation}
with $\rho_1=\frac12$, $\rho_2=\frac{2-\alpha}{2}$, $\sigma_1=2^{-1/3}$, and $\sigma_2=\frac{\alpha^{1/3}(6-10\alpha+6\alpha^2-\alpha^3)^{1/3}}{2(2-\alpha)}$.
\end{cor}

\medskip

\begin{cor}[At the $F_2$--$F_2$ shock]\label{Cor3b}
For a fixed $\beta\in (0,1)$, consider the initial condition given by $x_n(0)=-n-\lfloor\beta\ell\rfloor$ for $n\geq 1$ and $x_n(0)=-n$ for $-\lfloor\beta\ell\rfloor\leq n\leq 0$. Then, for $\alpha=1$, $\nu=\frac{(1-\beta)^2}{4}$ it holds
\begin{equation}
\lim_{\ell\to\infty}\Pb\left(x_{\nu \ell+\xi \ell^{1/3}}(t=\ell)\geq -s\ell^{1/3}\right) =F_2\left(\frac{s-\xi/\rho_1}{\sigma_1}\right)F_2\left(\frac{s-\xi/\rho_2}{\sigma_2}\right)
\end{equation}
with $\rho_1=\frac{1-\beta}{2}$, $\rho_2=\frac{1+\beta}{2}$, $\sigma_1=\frac{(1+\beta)^{2/3}}{2^{1/3}(1-\beta)^{1/3}}$, and $\sigma_2=\frac{(1-\beta)^{2/3}}{2^{1/3}(1+\beta)^{1/3}}$.
\end{cor}

As expected by KPZ universality, if we move away from the shock, the distribution function considered above becomes a single GOE or GUE distribution, with GOE whenever the particles density is constant and GUE whenever the particle density is decreasing, e.g., in the $F_2$-$F_2$ shock, the particle density is decreasing both to the left and to the right of the shock.

\section{Proof of Theorem~\ref{ThmMain}}\label{sectProofGeneric}
In the following, we will several times use the following two lemmas from~\cite{BC09}. By $``\Rightarrow"$ we designate convergence in distribution.
\begin{lem}[Lemma 4.1 in \cite{BC09}]\label{lemma4.1}
Let $D$ be a probability distribution and $(X_n)_{n\in \mathbb{N}}$ be a sequence of random variables.
If $X_n\geq \tilde{X}_n$ and $X_n \Rightarrow D$ and $X_n-\tilde{X}_n$ converges to zero in probability, then $\tilde{X}_n \Rightarrow D$ as well.
\end{lem}
\begin{lem}[Lemma 4.2 in \cite{BC09}]\label{lemma4.2} Let $(X_n)_{n\in \mathbb{N}}$, $(Y_n)_{n\in \mathbb{N}}$, $(\tilde{X}_n)_{n\in \mathbb{N}}$, $(\tilde{Y}_n)_{n\in \mathbb{N}}$ be sequences of random variables and $D_1,D_2,D_3$ be probability distributions.
Assume $X_n\geq \tilde{X}_n$ and $X_n \Rightarrow D_1$ as well as $\tilde{X}_n \Rightarrow D_1$; and similarly $Y_n\geq \tilde{Y}_n$ and $Y_n \Rightarrow D_2$ as well as $\tilde{Y}_n \Rightarrow D_2$. Let $Z_n=\max\{X_n,Y_n\}$ and \mbox{$\tilde{Z}_n=\max\{\tilde{X}_n, \tilde{Y}_n\}$}. Then if $\tilde{Z}_n\Rightarrow D_3$, we also have $Z_n \Rightarrow D_3$.
\end{lem}
We denote
\begin{equation}
L_{\mathcal{L}^{+}\to E}^{\mathrm{resc}}=\frac{L_{\mathcal{L}^+\to (\eta t, t)}-\mu  t}{ t^{1/3}},
\end{equation}
i.e. the last passage time $L_{\mathcal{L}^{+}\to E}$ rescaled as required by Assumption~\ref{Assumpt1},
we define analogously $L_{\mathcal{L}^{-}\to E}^{\mathrm{resc}}, L_{E^{+}\to (\eta t,t)}^{\mathrm{resc}}$ and $L_{\mathcal{L}^{+}\to E^{+}}^{\mathrm{resc}}$ as the last passage times rescaled as required by Assumption~\ref{Assumpt1} resp.~\ref{Assumpt2}.
We first note the following.
\begin{prop}\label{step1}
If $ \max\{L_{\mathcal{L}^{+}\to E}^{\mathrm{resc}},L_{\mathcal{L}^{-}\to E}^{\mathrm{resc}}\}\Rightarrow D$ as $t\to\infty$,
then
\begin{equation}
\frac{L_{\mathcal{L}\to E}-\mu t}{  t^{1/3}}\Rightarrow D.
\end{equation}
\end{prop}
\begin{proof}
Simply note that \mbox{$L_{\mathcal{L}\to E}=\max\{L_{\mathcal{L}^{+}\to E},L_{\mathcal{L}^{-}\to E}\}$}.
\end{proof}
 Thus it suffices to determine the limiting distribution of $\max\{L_{\mathcal{L}^{+}\to E}^{\mathrm{resc}},L_{\mathcal{L}^{-}\to E}^{\mathrm{resc}}\}$. We can actually reduce our problem a bit more.
\begin{prop}\label{step2}
Under Assumptions~\ref{Assumpt1} and~\ref{Assumpt2},
\begin{equation}
\max\left\{\frac{L_{\mathcal{L}^{+}\to E^{+}}+L_{E^{+}\to E}-\mu t}{ t^{1/3}},L_{\mathcal{L}^{-}\to E}^{\mathrm{resc}}\right\}\Rightarrow D \end{equation}
implies
\begin{equation}
\frac{L_{\mathcal{L}\to E}-\mu t}{t^{1/3}}\Rightarrow D.
\end{equation}
\end{prop}
\begin{proof}
We  have
\begin{equation}\label{downestimate}
L_{\mathcal{L}^{+}\to E}^{\mathrm{resc}}\geq \frac{L_{\mathcal{L}^{+}\to E^{+}}-\mu t+\mu_0 t^{\nu}}{t^{1/3}}+\frac{L_{E^{+}\to E}-\mu_0 t^{\nu}}{ t^{1/3}}=L_{\mathcal{L}^{+}\to E^{+}}^{\mathrm{resc}}+L_{ E^{+}\to (\eta t, t)}^{\mathrm{resc}}.
\end{equation}
By Assumption~\ref{Assumpt2}, $L_{\mathcal{L}^{+}\to E^{+}}^{\mathrm{resc}}$ converges to $G_1$.
Also by Assumption~\ref{Assumpt2},
$L_{E^{+}\to E}$ has fluctuations of order $t^{\nu/3}$, thus one can write
\begin{equation}\label{fadeaway}
L_{E^{+}\to E}^{\mathrm{resc}}=\frac{1}{t^{(1-\nu)/3}} X_t,
\end{equation}
where $X_t$ is a random variable converging to $G_0$.
In particular,
\eqref{fadeaway} vanishes as $t\to\infty$.
Applying Lemma~\ref{lemma4.2} to
\mbox{$X_n=L_{\mathcal{L}^{+}\to E}^{\mathrm{resc}}$}, \mbox{$\tilde{X_n}=\left(L_{\mathcal{L}^{+}\to E^{+}}+L_{E^{+}\to E}-\mu t\right)/ t^{1/3}$} and $Y_n=\tilde{Y}_n=L_{\mathcal{L}^{-}\to E}^{\mathrm{resc}}$ finishes the proof.
\end{proof}
Using the preceeding Propositions we can now prove Theorem~\ref{ThmMain}.
\begin{proof}[Proof of Theorem~\ref{ThmMain}]
Define, for some set $B$ and point $C$, $\tilde{L}_{B\to C}$ to be the last passage time of all paths from
$B$ to $C$ conditioned not to contain any point
$\bigcup_{\gamma \in [0,1-t^{\beta-1}]}D_{\gamma} $ with \mbox{$D_\gamma$} as in Assumption~\ref{Assumpt3}.
Then,
\begin{equation}
\begin{aligned}
\Pb\bigg(\bigg|\frac{L_{\mathcal{L}^{+}\to E^{+}}-\tilde{L}_{\mathcal{L}^{+}\to E^{+}}}{ t^{1/3}}\bigg|>\e\bigg)
&\leq
\Pb\bigg(\bigcup_{D_{\gamma}\atop \gamma \in [0,1-t^{\beta-1}]}\{D_\gamma\in \pi^{\rm max}_{\mathcal{L}^{+}\to E^{+}}\}\bigg)\to 0
\end{aligned}
\end{equation}
as $t\to\infty$, so that
\begin{equation}\label{conver1}
\Pb\bigg(\frac{\tilde{L}_{\mathcal{L}^{+}\to E^{+}}+\tilde{L}_{E^{+}\to E}-\mu t}{ t^{1/3}}\leq s\bigg)\to G_1(s)
\end{equation}
by the vanishing of~\eqref{fadeaway} and Lemma~\ref{lemma4.1}. Using Assumptions~\ref{Assumpt1} and~\ref{Assumpt3}, an analogous argument shows
\begin{equation}\label{conver2}
\Pb\big(\tilde{L}_{\mathcal{L}^{-}\to E}^{\mathrm{resc}}\leq s\big)\to G_2(s).
\end{equation}

 Let $\varepsilon >0$ and recall $X_t$ from \eqref{fadeaway}.
 We take $R>0$ such that with \mbox{$A_R=\{|\tilde{X}_t|<R\}$ $\Pb(A_{R}^{c})\leq \e$} for all $t$ large enough.
This implies that
\begin{equation}\begin{aligned}\label{check0}
\big|\Pb(\{\tilde{L}_{\mathcal{L}^{+}\to E^{+}}^{\mathrm{resc}}+t^{(\nu-1)/3}\tilde{X}_t\leq s\}\cap  A_R &\cap  \{\tilde{L}_{\mathcal{L}^{-}\to E}^{\mathrm{resc}}\leq s\})\\
-\Pb(\{\tilde{L}_{\mathcal{L}^{+}\to E^{+}}^{\mathrm{resc}}+t^{(\nu-1)/3}\tilde{X}_t\leq s\} &\cap \{\tilde{L}_{\mathcal{L}^{-}\to E}^{\mathrm{resc}}\leq s\})\big|\leq \e
\end{aligned}
\end{equation}
Then,

\begin{align}\label{check1}
&\,\,\Pb(\{\tilde{L}_{\mathcal{L}^{+}\to E^{+}}^{\mathrm{resc}}+t^{(\nu-1)/3}R\leq s\}\cap \{\tilde{L}_{\mathcal{L}^{-}\to E}^{\mathrm{resc}}\leq s\})-\e\\
&\leq\,\label{check2}\Pb(\{\tilde{L}_{\mathcal{L}^{+}\to E^{+}}^{\mathrm{resc}}+t^{(\nu-1)/3}\tilde{X}_t\leq s\}\cap A_R \cap \{\tilde{L}_{\mathcal{L}^{-}\to E}^{\mathrm{resc}}\leq s\})\\
&\leq\,\label{check3}\Pb(\{\tilde{L}_{\mathcal{L}^{+}\to E^{+}}^{\mathrm{resc}}-t^{(\nu-1)/3}R\leq s\}\cap A_R \cap \{\tilde{L}_{\mathcal {L}^{-}\to E}^{\mathrm{resc}}\leq s\})\\
&\leq\,\label{check4}\Pb(\{\tilde{L}_{\mathcal{L}^{+}\to E^{+}}^{\mathrm{resc}}-t^{(\nu-1)/3}R\leq s\}\cap  \{\tilde{L}_{\mathcal {L}^{-}\to E}^{\mathrm{resc}}\leq s\})
\end{align}

Finally, by construction, $\tilde{L}_{\mathcal{L}^{+}\to E^{+}}^{\mathrm{resc}}$  and $\tilde{L}_{\mathcal{L}^{-}\to E}^{\mathrm{resc}}$ are independent random variables, since $\beta< \nu$ and $\pi_{\mathcal{L}^{-}\to E}^{\mathrm{max}}$ has to pass to the right of $D_{1-t^{\beta-1}}$ by conditioning.
Due to  this independence, the fact that $\nu<1$ and the convergence in \eqref{conver1}, \eqref{conver2}, there is a $t_0$ such that for $t>t_0$
\begin{equation}
\begin{aligned}\label{ineq}
G_1(s) G_2(s)
-2\varepsilon &\leq \eqref{check1}\leq \eqref{check2}\leq\eqref{check4}\leq
G_1(s)G_2(s)
+\varepsilon.
\end{aligned}
\end{equation}
Thus applying \eqref{check0} to \eqref{check2} yields

\begin{equation}
\begin{aligned}
&\bigg|\Pb\left(\{\tilde{L}_{\mathcal{L}^{+}\to E^{+}}^{\mathrm{resc}}+t^{(\nu-1)/3}\tilde{X}_t\leq s\} \cap \{\tilde{L}_{\mathcal{L}^{-}\to E}^{\mathrm{resc}}\leq s\}\right)-
G_1(s)G_2(s)
\bigg|\leq 3\e,
\end{aligned}
\end{equation}
for all $t$ large enough. Therefore
\begin{equation}
\begin{aligned}
&\lim_{t\to\infty}\Pb\left(\max\bigg\{\frac{\tilde{L}_{\mathcal{L}^{+}\to E^{+}}+\tilde{L}_{E^{+}\to E}-\mu t}{ t^{1/3}},\tilde{L}_{\mathcal{L}^{-}\to E}^{\mathrm{resc}}\bigg\}\leq s\right)
=\, G_1(s) G_2(s).
\end{aligned}
\end{equation}
Applying Lemma~\ref{lemma4.2} to \mbox{$X_n=\left(L_{\mathcal{L}^{+}\to E^{+}}+L_{E^{+}\to E}-\mu t\right)/t^{1/3}$}, \mbox{$Y_n=L_{\mathcal{L}^{-}\to E}^{\mathrm{resc}}$}, \mbox{$\tilde{X}_n=\big(\tilde{L}_{\mathcal{L}^{+}\to E^{+}}+\tilde{L}_{E^{+}\to E}-\mu t\big)/t^{1/3}$}, $\tilde{Y}_n=\tilde{L}_{\mathcal{L}^{-}\to E}^{\mathrm{resc}}$, and using Proposition~\ref{step2} finishes the proof.
\end{proof}

\section{Results on specific LPP}\label{SectAllLPP}
In this section we derive some results on the LPP model with
\begin{equation}
\begin{aligned}
\omega_{i,j}\sim \exp(1),\quad & j\geq 1,\\
\omega_{i,j}\sim \exp(\alpha),\quad & j\leq 0,
\end{aligned}
\end{equation}
and with two half-lines given by
\begin{equation}
\mathcal{L}^{+}=\{(-v,v)|v \in \Z_+\}\textrm{ and }\mathcal{L}^{-}=\{(-v,v)|v \in\Z_-\}.
\end{equation}
Assumptions~\ref{Assumpt1}-\ref{Assumpt2} will be verified by using the results of Section~\ref{sectLPP}.
After that, in Section~\ref{sectProofNocrossing} we determine the no-crossing results corresponding to Assumption~\ref{Assumpt3}.

\subsection{Deviation Results for LPP}\label{sectLPP}

\subsubsection{Point-to-point LPP results}
First we remind a result of Johansson (Theorem 1.6 of~\cite{Jo00b}, originally stated for $\eta\geq 1$, but by symmetry of the LPP one easily extends the statement for any $\eta>0$).
\begin{prop}[Point-to-point LPP: convergence to $F_2$]\label{propJohConvergence}
Let $0<\eta<\infty$. Then,
\begin{equation} \begin{aligned}
\lim_{\ell\to\infty}\Pb\left(L_{0\to (\lfloor\eta\ell\rfloor,\lfloor\ell\rfloor)}\leq \mu_{\rm pp}\ell +s \sigma_\eta \ell^{1/3}\right)= F_2(s)
\end{aligned}\end{equation}
where $\mu_{\rm pp}=(1+\sqrt{\eta})^2$, and $\sigma_\eta=\eta^{-1/6}(1+\sqrt{\eta})^{4/3}$, and $F_2$ is the GUE Tracy-Widom distribution function.
\end{prop}

The distribution function of $L_{0\to (\lfloor\eta\ell\rfloor,\lfloor\ell\rfloor)}$ has the following known decay\footnote{One could improve the decay of Proposition~\ref{devone} to $\exp(-c s^{3/2})$ and of Proposition~\ref{devtwo} to $\exp(-c |s|^3)$, but it is not needed for our purposes.}.
\begin{prop}[Point-to-point LPP: upper tail]\label{devone}
Let $0<\eta<\infty$. Then for given $\ell_0>0$ and $s_0 \in \mathbb{R}$, there exist constants $C,c>0$ only dependent on $\ell_0,s_0$ such that for all $\ell\geq \ell_0$ and $s\geq s_0$ we have
\begin{equation} \begin{aligned} \label{moddef}
\Pb\left(L_{0\to (\lfloor\eta\ell\rfloor,\lfloor\ell\rfloor)}> \mu_{\rm pp}\ell +\ell^{1/3}s\right)\leq C \exp(-c s),
\end{aligned}\end{equation}
where $\mu_{\rm pp}=(1+\sqrt{\eta})^{2}$.
\end{prop}
\begin{proof}
By symmetry, it is enough to consider $\eta\in (0,1]$. Also, we will (re)derive the statement for the complementary event.
As stated in Proposition 6.1 of~\cite{BBP06}, we  have
\begin{equation}
\Pb(\lambda_{1}(m-d,m+d) \leq u)=\Pb\left(L_{0\to (\lfloor\eta\ell\rfloor,\lfloor\ell\rfloor)}\leq u \right),
\end{equation}
where $\lambda_1$ is the largest eigenvalue of a \mbox{$(m-d)\times(m+d)$} Laguerre Unitary Ensemble (LUE), i.e., the largest eigenvalue of
$\frac{1}{m-d}XX^{*}$, where $X$ is a \mbox{$(m-d)\times(m+d)$} matrix with i.i.d.\ standard complex  Gaussian entries; the choice of parameters is
 so that $m+d=\lfloor\eta\ell/\mu_{\rm pp}\rfloor$ and $m-d=\lfloor\ell/\mu_{\rm pp}\rfloor$ (explicitly, one might take
$m=\lfloor\frac{\ell(\eta+1)}{2\mu_{\rm pp}}\rfloor$ and $d=\lfloor\frac{\ell(1-\eta)}{2\mu_{\rm pp}}\rfloor$, but then these identites
might only hold with an error $\pm 1$).
Take $K_{m,d}$ to be the kernel (3.13) of~\cite{FS05a} (with $w=0$), which, according to Proposition C.1 of~\cite{FS05a}, is a conjugated correlation kernel for the LUE. Then, with  $\chi_{u}=\Id_{(u,+\infty)}$
\begin{equation} \begin{aligned}\label{equalities}
F(u):= \det(1-\chi_{u}K_{m,d}\chi_u)=\Pb(\lambda_{1}(m-d,m+d) \leq u)
\end{aligned}\end{equation}
 Define the function $u(s,\ell)=\ell-s\ell^{1/3}$. The decay of $F(u)$ is known, see
 (37) in~\cite{BFP12}; more precisely  we have with $C,d>0$ dependent on $s_0 \in \mathbb{R}$ and $\ell_0>0$
\begin{equation}
1-Ce^{-ds}\leq F(u(s,\ell))
\end{equation} for $\ell>\ell_0$ and  $s>s_0$. Making the change of variable $\ell\to \mu_{\rm pp}\ell $,   (\ref{moddef}) follows with $c=d/\mu_{\rm pp}^{1/3}$.
\end{proof}
\begin{prop}[Point-to-point LPP: lower tail]\label{devtwo} Let $0<\eta<\infty$ and $\mu_{\rm pp}=(1+\sqrt{\eta})^2$. There exist positive constants $s_0,\ell_0,C,c$
such that for $s\leq -s_0,$ $\ell\geq \ell_0$,
\begin{equation} \begin{aligned}
\Pb\left(L_{0\to (\lfloor\eta\ell\rfloor,\lfloor\ell\rfloor)}\leq \mu_{\rm pp} \ell +s\ell^{1/3}\right)\leq C\exp(-c|s|^{3/2}).
\end{aligned}\end{equation}
\end{prop}

\begin{proof}
Take the functions $F,u(s,t)$ and the parameters $w,m,d$ as in the proof of Proposition~\ref{devone}.
Proposition 3 of~\cite{BFP12} (to be found in the proof of Proposition 2 of~\cite{BFP12}) and the inequality (56) of the same paper imply that there exist positive constants $s_0,t_0,C,c$ such that
\begin{equation}
F(u(s,t))\leq C \exp(-c |s|^{3/2}),
\end{equation}
for all $s\leq -s_0$ and $t\geq t_0$.
\end{proof}

\subsubsection{Half-line $\mathcal{L}^+$-to-point LPP results}
To obtain the results for the LPP from the half-line $\mathcal{L}^+$ to a point $(\eta \ell,\ell)$, we use the correspondence of LPP and TASEP, namely
\begin{equation}\label{eqLPPtasep}
\Pb\left(L_{\mathcal{L}^{+}\to (m,n)}\leq t\right)=\Pb\left(x_n(t)+n\geq m\right),
\end{equation}
where $x_n(t)$ is the position at time $t$ of the TASEP particle that started from $x_n(0)=-2n$ in the initial configuration where particles occupy $-2\N_0$. TASEP particle have all jump rate $1$.
The latter distribution function is expressed as a Fredholm determinant of a kernel $\hat{K}_{n,t}$, as is shown in~\cite{BFS07}.
\begin{prop}[Proposition 3 in~\cite{BFS07}]\label{PropTASEPHalfFlat}
Let particle number $n \in \mathbb{N}_0$ start in $-2n$ at time $t=0$. Denote by $x_n (t)$ the position of particle number $n$ at time $t$.
We then have
\begin{equation}
\Pb(x_n(t) > s)=\det(1-\chi_s \hat{K}_{n,t} \chi_s)_{\ell^2(\Z)}
\end{equation}
where $\chi_s =\Id_{(-\infty,s]}$ and $\hat{K}_{n,t}$ is given by\footnote{For a set $S$, the notation $\Gamma_{S}$ means a path anticlockwise oriented enclosing only poles of the integrand belonging to the set $S$.}
\begin{equation}\label{kernelhalfflat}
\begin{aligned}
&\hat{K}_{n,t}(x_1,x_2)
=\frac{1}{(2\pi i)^{2}} \oint_{\Gamma_1}\dx v\oint_{\Gamma_{0,1-v}}\frac{\dx w}{w}\frac{e^{tw}(w-1)^{n }}{w^{x_1 + n}}\frac{v^{x_2 + n }}{e^{tv}(v-1)^{n}}\\
&\hspace{17em}\times\frac{2v-1}{(w+v-1)(w-v)}.
\end{aligned}
\end{equation}
\end{prop}

To get a bound for the upper tail we need to have the following estimate of the decay of the kernel.
\begin{prop}[Exponential decay $\hat{K}_{n,t}$]\label{propDecayHalfLinePt}
Consider the scaling
\begin{equation}\label{moddevscaling}
n(t)=\left[\frac{r}{4}t\right] \quad x_{i}=\left[\frac{1-r}{2}t-s_{i}t^{1/3}\right],
\end{equation}
for some $r>1$.
With this choice, there exists a constant $C$ and a $t_0$ such that for $t>t_0$ and $s_{1},s_{2} \geq 0$
\begin{equation}\label{expdecay}
|\hat{K}_{n,t} (x_1,x_2)t^{1/3}2^{x_2-x_1} e^{-(s_2-s_1)/2}|\leq C\, e^{-(s_1+s_2)/2}.
\end{equation}
\end{prop}
\begin{proof}[Proof of Proposition~\ref{propDecayHalfLinePt}]
Below we will show that for $t$ large enough, there are
 constants $C,\mu(r)>0$ such that  we have uniformly in $s_1,s_2\geq 0$
\begin{equation}
|\hat{K}_{n,t} (x_1,x_2)t^{1/3}2^{x_2-x_1} |\leq C e^{-(s_1+s_2)}+C t^{1/3} e^{-\mu(r) t} e^{s_1 t^{1/3}\ln(2-r)}.
\end{equation}
From this then follows that
\begin{equation}
|\hat{K}_{n,t} (x_1,x_2)t^{1/3}2^{x_2-x_1} e^{-(s_2-s_1)/2}|\leq 2C e^{-(s_1+s_2)/2}
\end{equation}
since $t^{1/3} e^{-\mu(r) t}\leq 1$ and $e^{s_1 (t^{1/3}\ln(2-r)+1/2)}\leq 1$ for $t$ large enough (because $\ln(2-r)<0$).

Therefore, below we need to bound $\hat{K}_{n,t} (x_1,x_2)t^{1/3}2^{x_2-x_1}$. We can divide the kernel $\hat K_{n,t}$ into the contribution coming from the residue at $w=-v+1$ and the rest. The contribution of this residue is
\begin{equation}\label{residue1}
(-1)^{x_1+1}2^{x_2-x_1} \frac{t^{1/3}}{2\pi i} \oint_{\Gamma_{1}}\text{d}v\frac{v^{x_2+2n}}{(1-v)^{x_1+2n+1}}e^{(1-2v)t}
\end{equation}
This kernel was already analyzed in~\cite{BF07}. Indeed, (\ref{residue1}) is the kernel from Proposition 5.3 in~\cite{BF07} for the special choice of parameters $t_1=t_2=T=t$, $L=0$, and $R=1$.
Our scaling also fits in the one from  (2.9) in~\cite{BF07};
take
$\pi(\theta)=r/4+\theta$ and  $\theta$ to be the solution of $r/4+2\theta=1$, i.e., $\theta=1/2-r/8$. Then (2.9) in~\cite{BF07} equals exactly (\ref{moddevscaling}).
Said Proposition yields now that for any $(s_1,s_2) \in [-l,\infty)^{2}$ we have
\begin{equation} \begin{aligned}
|(\ref{residue1})| \leq \text{const}\, e^{-(s_1+s_2)}.
\end{aligned}\end{equation}
Let us deal now with the remaining part. Taking $\tilde{s}_{i}=s_i t^{-2/3}$, we have to bound the kernel

\begin{equation}\label{eq23}
\begin{aligned}
2^{x_2-x_1}&\frac{t^{1/3}}{(2\pi i)^{2}}\oint_{\Gamma_1}\text{d}v\oint_{\Gamma_{0}}\frac{\text{d}w}{w}\frac{e^{tw}(w-1)^{n}}{w^{x_1 + n}}
\frac{v^{x_2 + n}}{e^{tv}(v-1)^{n}}\frac{2v-1}{(w+v-1)(w-v)}
\\=&\frac{t^{1/3}}{(2\pi i)^{2}}\oint_{\Gamma_1}\text{d}v\oint_{\Gamma_{0}}\text{d}w
\frac{e^{tf_{0}(w,\tilde s_1)}}{e^{tf_{0}(v,\tilde s_2)}}
\frac{2v-1}{(w+v-1)(w-v)}
\end{aligned}
\end{equation}
with
\begin{equation}\label{myf}
f_{0}(w,s)=\frac{r}{4}\ln(w-1)+w-\frac{2-r}{4}\ln(w)+s\ln(2w).
\end{equation}
We first note that for $r\geq 2$ the pole at $w=0$ disappears and thus (\ref{eq23}) vanishes. We therefore assume $1<r<2$ in the following.
We now claim that
\begin{equation}
\Gamma_{0}(t)=\lambda e^{it},\quad t \in [0,2\pi)
\end{equation}
is a steep descent path of $f_{0}$ for $\lambda =1-r/2$.
To check the steep descent condition, note
\begin{equation}
\begin{aligned}
\Re(f_{0}(\Gamma_0 (t)&,\tilde{s}_1))
=\tilde{s}_1\ln(2\lambda)+\lambda \cos(t)-\frac{2-r}{4}\ln(\lambda)
+\frac{r}{4}\ln(|\lambda e^{it}-1|)
\\&=\tilde{s}_1\ln(2\lambda)+\lambda \cos(t)-\frac{2-r}{4}\ln(\lambda)
+\frac{r}{8}\ln\left(\lambda^{2}+1-2\lambda\cos(t)\right).
\end{aligned}
\end{equation}
Thus we have
\begin{equation}\label{whenzerooo}
\frac{\partial}{\partial t}
\Re\left(f_{0}(\Gamma_0 (t),\tilde{s}_1)\right)
=-\lambda\sin(t)\left(1 -\frac{r/4}{|\lambda e^{it}-1|^2}\right),
\end{equation}
which is strictly negative for all $t\in (0,\pi)$ (and strictly positive for \mbox{$t\in (\pi,2\pi)$}). Indeed, $|\lambda e^{it}-1|\geq r/2$, from which $1 -\frac{r/4}{|\lambda e^{it}-1|^2}\geq 1-1/r>0$. Thus $\Gamma_0$ as chosen above is a steep descent path for $f_0$ with maximum at $t=0$.

For $\Gamma_{1}$, we choose
\begin{equation}
\Gamma_{1}(t)=1-\frac{1}{2} e^{it}, \quad t \in [0,2\pi)
\end{equation}
and we want to show that it is a steep descent path for $-f_{0}$.
We have
\begin{align*}
\Re(-f_{0}(\Gamma_{1}(t),\tilde{s}_2))=&-\frac{r}{4}\ln(1/2)+\frac{2-r}{8}\ln\left(5/4-\cos(t)\right)+\frac12\cos(t)\\
&-\tilde{s}_2\ln(|2-e^{it}|).
\end{align*}
The term $-\tilde{s}_2\ln(|2-e^{it}|)$ reaches clearly its maximum at $t=0$ for any $\tilde s_2\geq 0$. Thus we can focus on the $\tilde s_2=0$ case. We have
\begin{equation}\label{scderiv}
\frac{\partial}{\partial t} \Re\left(f_{0}(\Gamma_1(t),0)\right)
=-\frac{\sin(t)}{2} \left(1-\frac{2-r}{8}\frac{1}{|1-\frac12 e^{it}|^2}\right),
\end{equation}
which is strictly negative for $t\in(0,\pi)$ and strictly positive for $t\in(\pi,2\pi)$. This follows from $|1-\frac12 e^{it}|\geq 1/2$, so that $1-\frac{2-r}{8}|1-\frac12 e^{it}|^{-2}\geq r/2>0$. Thus $\Gamma_1$ is a steep descent path for $-f_0$ attaining its maximum at $t=0$.

The paths $\Gamma_0$ and $\Gamma_1$ are such that the factor $\frac{2v-1}{(w+v-1)(w-v)}$ in (\ref{eq23}) is uniformly bounded and the length of the paths is also bounded. Therefore, since $\Gamma_0$ and $\Gamma_1$ are steep descent paths, we get the easy bound
\begin{equation}\label{eq28}
\begin{aligned}
|(\ref{eq23})|\leq t^{1/3}e^{t(f_{0}(1-r/2,\tilde{s}_1)-f_{0}(1/2,\tilde{s}_2))}=t^{1/3}e^{-\mu(r) t}e^{t^{1/3}\ln(2-r)s_1},
\end{aligned}
\end{equation}
with $\mu(r)=-\frac{r}{4}\ln(r)-\frac{1-r}{2}+\frac{2-r}{4}\ln(2-r)>0$ for all $1<r<2$.
\end{proof}

\begin{prop}\label{propHalfLineBound}
Fix an $0< \eta < 1$ and let $\mu=2(1+\eta)$. Then, for any $\e\in [0,2(1-\eta))$, there exists constants $C,\tilde{c}>0$ and $\ell_0 >0$ such that for all $\ell>\ell_0$
\begin{equation} \begin{aligned}\label{devest}
\Pb\left(L_{\mathcal{L}^{+}\to (\lfloor\eta\ell\rfloor,\lfloor\ell\rfloor)}> (\mu+\e/2) \ell\right) \leq C\exp\left(-\tilde{c} \,\e \ell^{2/3}\right).
\end{aligned}\end{equation}
\end{prop}
\begin{proof}[Proof of Proposition~\ref{propHalfLineBound}]
We follow along the lines of the proof of Theorem~2.5 in Section 5 of~\cite{BFP06}. We use the relation (\ref{eqLPPtasep}) between LPP and TASEP, in which we set $t:=(\mu+\e/2) \ell$ and denote by $\ell(t)=t/(\mu+\e/2)$ its inverse function. Then, using this relation and Proposition~\ref{PropTASEPHalfFlat}, we see that
\begin{equation}\label{eq12}
(\ref{devest})=1-\Pb\left(x_{\ell(t)}(t)\geq (\eta-1)\ell(t)\right).\\
\end{equation}
Let us denote
\begin{equation}
X^{\rm resc}_t=\frac{x_{\ell(t)}(t)-(-2 \ell(t)+t/2)}{-t^{1/3}}.
\end{equation}
Then,
\begin{equation}\label{yes}
\begin{aligned}
(\ref{eq12})&=1-\Pb\left(X^{\rm resc}_t \leq \frac{(\eta+1)\ell(t)-t/2}{-t^{1/3}}  \right)\\
&=-\sum_{m=1}^{\infty}\frac{(-1)^{m}}{m!}\int \text{d}s_1\cdots\int \text{d}s_m  \det[t^{1/3}\hat{K}_{\ell(t),t}([x(s_i)],[x(s_j)])]_{1\leq i,j\leq m}
\end{aligned}
\end{equation}
where $x(s)=(-2\ell(t)+t/2)-s t^{1/3}$ and the integration domain of the $s_i$'s is $(\e t^{2/3}/4(\mu+\e/2),\infty)$.
On (\ref{yes}) we apply Proposition~\ref{propDecayHalfLinePt} with \mbox{$r=4/(\mu+\e/2)$}.

We can thus single out a product $\prod_{i=1}^{m}e^{-s_i}$ of the determinant, so that the absolute value of all entries in the matrix  is bounded by a constant $C$, so using Hadamard's bound, we get
\begin{equation} \begin{aligned}
|(\ref{yes})|&\leq \sum_{m=1}^{\infty}\frac{C^{m}m^{m/2}}{m!}\int_{\e t^{2/3}/4(\mu+\e/2)}\text{d}s_1\cdots\int_{\e t^{2/3}/4(\mu+\e/2)} \text{d}s_{m}\prod_{i=1}^{m}e^{-s_i}
\\&=\sum_{m=1}^{\infty}\frac{(2C)^{m}m^{m/2}\exp\left(-m \e t^{2/3}/4(\mu+\e/2)\right)}{m!}\\&\leq \tilde{C}\exp\left(-\e t^{2/3}/4(\mu+\e/2)\right)\leq \tilde{C}\exp\left(-\tilde c \e \ell^{2/3}\right)
\end{aligned}\end{equation}
for some constants $\tilde C,\tilde c$ (uniform in $\ell$).
\end{proof}

\begin{prop}[Half-line $\mathcal{L}^+$-to-point LPP: convergence to $F_1$]\label{propHalfFlatConvergence}
For any fixed $0< \eta < 1$, it holds
\begin{equation} \begin{aligned}
\lim_{\ell\to\infty}\Pb\left(L_{\mathcal{L}^{+}\to (\lfloor\eta\ell\rfloor,\lfloor\ell\rfloor)}\leq \mu\ell +s \tilde \sigma_\eta \ell^{1/3}\right)= F_1(2s)
\end{aligned}\end{equation}
where $\mu=2(1+\eta)$, $\tilde \sigma_\eta=2^{4/3}(1+\eta)^{1/3}$, and $F_1$ is the GOE Tracy-Widom distribution function.
\end{prop}
\begin{proof}[Proof of Proposition~\ref{propHalfFlatConvergence}]
As in the proof of Proposition~\ref{propHalfLineBound} we use the relation (\ref{eqLPPtasep}) between LPP and TASEP, in which we set $t:=\mu \ell+s\tilde \sigma_\eta \ell^{1/3}$ and denote by
\begin{equation}
\ell(t)=\frac{t}{\mu}-2s \frac{t^{1/3}}{\mu}+o(1)
\end{equation}
its inverse function. Thus,
\begin{equation}\label{eq32}
\Pb\left(L_{\mathcal{L}^{+}\to (\lfloor\eta\ell\rfloor,\lfloor\ell\rfloor)}\leq \mu\ell +s \tilde \sigma_\eta \ell^{1/3}\right) =\Pb\left(x_{\ell(t)}(t)\geq (\eta-1)\ell(t) \right).
\end{equation}
Let us denote
\begin{equation}
X^{\rm resc}_t=\frac{x_{\ell(t)}(t)-(-2 \ell(t)+t/2)}{-t^{1/3}}.
\end{equation}
Then,
\begin{equation}
\begin{aligned}
(\ref{eq32}) &=\Pb\left(X^{\rm resc}_t \leq \frac{(\eta+1)\ell(t)-t/2}{-t^{1/3}} \right)\\
&=\sum_{m=0}^{\infty}\frac{(-1)^{m}}{m!}\int_s^\infty \text{d}s_1\cdots\int_s^\infty \text{d}s_m  \det[t^{1/3}\hat{K}_{\ell(t),t}([x(s_i)],[x(s_j)])]_{1\leq i,j\leq m}
\end{aligned}
\end{equation}
where $x(s)=(-2\ell(t)+t/2)-s t^{1/3}$. The bound of Proposition~\ref{propDecayHalfLinePt} allows us to apply dominated convergence and take the $t\to\infty$ (i.e., $\ell\to\infty$) inside the Fredholm series. Thus it remains to show that the rescaled kernel $t^{1/3}\hat{K}_{\ell(t),t}([x(s_i)],[x(s_j)])$, or a conjugation of it, converges pointwise to the Airy$_1$ kernel $\mathcal{A}_1(s_i,s_j)=\Ai(s_i+s_j)$.

As in Proposition~\ref{propDecayHalfLinePt}, we consider the kernel conjugated by the factor $2^{x(s_j)-x(s_i)}$. We can divide the kernel $\hat K_{n,t}$ into the contribution coming from (a) the residue at $u=-v+1$ and (b) the rest. The contribution coming from the residue is (\ref{residue1}), that is, the kernel for the flat initial configuration (all even sites are initially occupied by a particle). It was shown in Theorem~2.3 of~\cite{BFPS06} (see also Proposition~5.1 of~\cite{BF07}) that the kernel converges pointwise to the Airy$_1$ kernel. The control of the contribution of (b) is already made in the proof of Proposition~\ref{propDecayHalfLinePt}. Indeed, the estimate (\ref{eq28}) implies that this contribution goes to $0$ as $t\to\infty$ for all fixed $s\in\R$. This ends the proof of Proposition~\ref{propHalfFlatConvergence}, since $\det(\Id-\mathcal{A}_1)_{L^2(s,\infty)}=F_1(2s)$ by~\cite{FS05b}.
\end{proof}

A simple corollary of Proposition~\ref{propHalfLineBound} adapted to the problem we are looking at is the following.
\begin{cor}\label{CorLplus}
Fix an $0< \eta < 1$, a $\beta\in (1/3,1]$ and define
\begin{equation}
\gamma \in [0,1-t^{\beta-1}],\quad \e  = t^{-\chi}\textrm{ with } \chi\in (0,2/3).
\end{equation}
Then there exists constants $C,\tilde{c}>0$ and $t_0 >0$ such that for all $t >t_0$
\begin{equation} \begin{aligned}\label{devestBB}
\Pb\left(L_{\mathcal{L}^{+}\to D_\gamma}> \left(\mu_{\gamma}+\frac{\e }{2}\right)t\right) \leq C\exp\left(-\tilde{c} \,t^{2/3-\chi}\right).
\end{aligned}\end{equation}
\end{cor}
\begin{proof}
It is a straightforward consequence of Proposition~\ref{propHalfLineBound}. Indeed, setting $\ell=\gamma t$,
\begin{equation}
\begin{aligned}
\Pb\left(L_{\mathcal{L}^{+}\to D_\gamma}> \left(\mu_{\gamma}+\e/2\right)t\right) &=
\Pb\left(L_{\mathcal{L}^{+}\to (\lfloor\eta\ell\rfloor,\lfloor\ell\rfloor)}> (\mu+\e/(2\gamma)) \ell\right) \\
&\leq \Pb\left(L_{\mathcal{L}^{+}\to (\lfloor\eta\ell\rfloor,\lfloor\ell\rfloor)}> (\mu+\e/2) \ell\right)
\end{aligned}
\end{equation}
since $\gamma\in [0,1]$. Then the result is the bound (\ref{devest}).
\end{proof}

\subsubsection{Half-line $\mathcal{L}^-$-to-point LPP results}
To obtain the results for the LPP from the half-line $\mathcal{L}^-$ to a point $(\eta \ell,\ell)$, we use the correspondence of LPP and TASEP, namely
\begin{equation}\label{eqLPPtasepB}
\Pb\left(L_{\mathcal{L}^{-}\to (m,n)}\leq t\right)=\Pb\left(x_n(t)+n\geq m\right),
\end{equation}
where $x_n(t)$ is the position at time $t$ of the TASEP particle with label $n$. The initial condition is
\begin{equation}\label{eq43}
x_n(0)=-n, n\geq 1,\quad x_n(0)=-2n, n\leq 0,
\end{equation}
and the jump rates $v_n$ of particles are given by
\begin{equation}\label{eq44}
v_n=1,n\geq 1,\quad v_n=\alpha, n\leq 0.
\end{equation}

\begin{prop}\label{PropDistrLminus}
Let us consider TASEP with jump rates (\ref{eq44}) and initial condition (\ref{eq43}). Denote $x_n(t)$ the position of particle number $n$ at time $t$.
We then have
\begin{equation}\label{eq49}
\Pb(x_n(t) > s)=\det(1-\chi_s \tilde{K}_{n,t} \chi_s)_{\ell^2(\Z)}
\end{equation}
where $\chi_s =\Id_{(-\infty,s]}$ and $\tilde{K}_{n,t}=K_{n,t}^{(1)}+ K_{n,t}^{(2)}$ with
\begin{equation}\label{kernelhalfflatLminus}
\begin{aligned}
  K_{n,t}^{(1)}(x_1,x_2)&=\frac{1}{(2\pi \I)^{2}} \oint_{\Gamma_{-1}}\frac{\mathrm{d}w}{w+1}\oint_{\Gamma_{0,\alpha-2-w}}\mathrm{d}z\frac{e^{t(w+1)}w^{n}}{(w+1)^{x_1+n}}\\
  &\hspace{11em}\times\frac{(z+1)^{x_2+n}}{e^{t(z+1)}z^{n}}\frac{1}{z-(\alpha-2-w)},\\
 K_{n,t}^{(2)}(x_1,x_2)&=\frac{1}{(2\pi \I)^{2}}\oint_{\Gamma_{0}}\mathrm{d}z
 \oint_{\Gamma_{-1}}\frac{\mathrm{d}w}{w+1}\frac{e^{t(w+1)}w^{n}}{(w+1)^{x_1+n}}\frac{(z+1)^{x_2+n}}{e^{t(z+1)}z^{n}}\frac{1}{w-z}.
\end{aligned}
\end{equation}
\end{prop}
The proof of this proposition is not so short and it is given in Section~\ref{sectDistrLminus} below.

Next we show the point-wise convergence and get bounds for the properly rescaled kernel. Consider the scaling\footnote{Below we will no longer   write explicitly the integer values.}
\begin{equation}\label{ScalingLminus}
n=\left[\frac{\kappa (2-\alpha)}{4}t\right] \quad x_{i}=\left[\frac{\alpha-\kappa}{2}t-s_{i}t^{1/3}\right],
\end{equation}
for $\alpha\in [0,1)$ and $\kappa\in [0,1)$. Then, we define the rescaled and conjugated kernels by
\begin{equation}
K^{(i)}_{t,{\rm resc}}(s_1,s_2)=t^{1/3}(\alpha/2)^{x_1-x_2}K^{(i)}_{n,t}(x_1,x_2),\quad i=1,2,
\end{equation}
with $x_i$ and $n$ as in (\ref{ScalingLminus}). Before stating the results, let us manipulate the kernel slightly. Denote by $\tilde s_i=s_i t^{-2/3}$. In particular, we can assume \mbox{$0\leq \tilde s_1\leq \alpha(2-\kappa)/4$}, since otherwise the kernel is identically equal to zero. Because of that, the Fredholm determinant in (\ref{eq49}) is identically equal to zero for $s>\alpha(2-\kappa)t^{2/3}/4$. Therefore, below we can restrict our estimates to $s_1,s_2\leq \alpha(2-\kappa)t^{2/3}/4$ only.

Let us introduce the function
\begin{equation}\label{eq52}
f_0(w,\tilde s)=w+1+\frac{\kappa(2-\alpha)}{4}\ln(w)- \left(\frac{\alpha(2-\kappa)}{4}-\tilde s\right)\ln(2(w+1)/\alpha).
\end{equation}
we have
\begin{equation}
K^{(2)}_{t,{\rm resc}}(s_1,s_2) = \frac{t^{1/3}}{(2\pi \I)^{2}} \oint_{\Gamma_{0}}\mathrm{d}z \oint_{\Gamma_{-1}}\frac{\mathrm{d}w}{w+1}
\frac{e^{t f_0(w,\tilde s_1)}}{e^{t f_0(z,\tilde s_2)}}
\frac{1}{w-z}
\end{equation}
and, separating the contribution of the simple pole at $z=\alpha-2-w$ in $K^{(1)}_{n,t}$,
\begin{equation}
K^{(1)}_{t,{\rm resc}}(s_1,s_2) = K^{(1,a)}_{t,{\rm resc}}(s_1,s_2)+K^{(1,b)}_{t,{\rm resc}}(s_1,s_2)
\end{equation}
where
\begin{equation}\label{eq54}
\begin{aligned}
K^{(1,a)}_{t,{\rm resc}}(s_1,s_2) &= \frac{t^{1/3}}{(2\pi \I)^{2}} \oint_{\Gamma_{-1,\alpha-2}}\frac{\mathrm{d}w}{w+1}\oint_{\Gamma_{0}}\mathrm{d}z
\frac{e^{t f_0(w,\tilde s_1)}}{e^{t f_0(z,\tilde s_2)}}
\frac{1}{z-(\alpha-2-w)},\\
 K^{(1,b)}_{t,{\rm resc}}(s_1,s_2)&= \frac{t^{1/3}}{2\pi \I} \oint_{\Gamma_{-1,\alpha-2}}\frac{\mathrm{d}w}{w+1} e^{t [f_0(w,\tilde s_1)-f_0(\alpha-2-w,\tilde s_2)]}.
\end{aligned}
\end{equation}

\begin{figure}
\begin{center}
\psfrag{w}[c][b]{$w\in \Gamma_{-1,\alpha-2}$}
\psfrag{z}[c][b]{$z\in\Gamma_0$}
\psfrag{0}[c][b]{$0$}
\psfrag{m1}[c][b]{$-1$}
\psfrag{am2}[c][b]{$\alpha-2$}
\includegraphics[height=4cm]{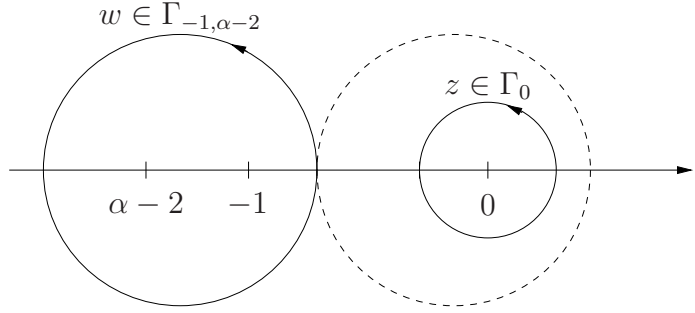}
\caption{Illustration of the paths used in the kernel $K_{t,\mathrm{resc}}^{(1,a)}$. The dashed line is the image of $\alpha-2-w$.}
\label{FigPaths1}
\end{center}
\end{figure}

\begin{remark}
$\alpha-2$ is not a pole for the double integral, but the reason why we have chosen the path for $w$ to encircle also $\alpha-2$ is the following. The function $-f_0(\alpha-2-w,\tilde s_2)$ has a pole at $w=\alpha-2$. Therefore, if, before computing the residue at $z=\alpha-2-w$, we choose the path $w$ so that it goes around $\alpha-2$ too, then, its image by $\alpha-2-w$ goes around the origin too, see Figure~\ref{FigPaths1}. This means that, the path for $z$ in the first term of (\ref{eq54}) will have to be chosen to stay inside the image of $\alpha-2-w$. We could have also chosen to have $\alpha-2$ outside the path for $w$, but this is not adequate to get the bounds on the kernel.
\end{remark}

\begin{remark}
For large $|w|$, the leading term in $f_0(w,\tilde s)$ is given simply the linear term $w$. So, we can as well consider (open) contours $\Gamma_{-1,\alpha-2}$ such that the real part of $w$ goes to $-\infty$, and similarly $\Gamma_0$ such that the real part of $z$ goes to $\infty$, see Figure~\ref{FigPaths2}.
\end{remark}
\begin{figure}
\begin{center}
\psfrag{w}[l][b]{$w\in \Gamma_{-1,\alpha-2}$}
\psfrag{w2}[r][b]{$\alpha-2-w$}
\psfrag{z}[l][b]{$z\in\Gamma_0$}
\psfrag{0}[c][b]{$0$}
\psfrag{m1}[c][b]{$-1$}
\psfrag{am2}[c][b]{$\alpha-2$}
\psfrag{k}[c][b]{$-\frac{\kappa}{2}$}
\psfrag{wc}[c][b]{$-1+\frac{\alpha}{2}$}
\includegraphics[height=4cm]{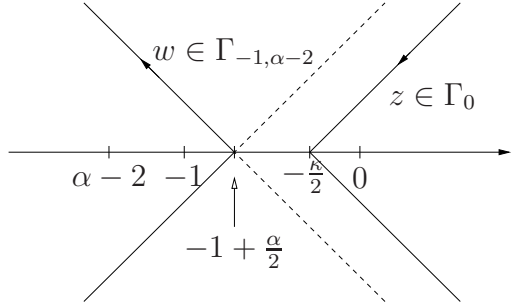}
\caption{Paths used for the asymptotic analysis in Proposition~\ref{propDecayHalfLinePt2a} and Proposition~\ref{propConvHalfLinePt2}. The dashed line is the image of $\alpha-2-w$.}
\label{FigPaths2}
\end{center}
\end{figure}

\begin{prop}[Bounds for $K^{(1,a)}_{t,{\rm resc}}$ and $K^{(2)}_{t,{\rm resc}}$]\label{propDecayHalfLinePt2a}
For any $\ell_0>0$, there exists a $t_0$ such that for $t>t_0$ and $s_{1},s_{2} \in [-\ell_0,\frac{\alpha(2-\kappa)}{4} t^{2/3}]$,
\begin{equation}
\begin{aligned}
|K^{(1,a)}_{t,{\rm resc}}(s_1,s_2)| &\leq e^{-t F(\alpha,\kappa)/2},\\
|K^{(2)}_{t,{\rm resc}}(s_1,s_2)| &\leq e^{-t F(\alpha,\kappa)/2},
\end{aligned}
\end{equation}
where
\begin{equation}
F(\alpha,\kappa)=-\frac{\alpha+\kappa-2}{2}-\frac{\kappa (2-\alpha)}{4}\ln\left(\frac{2-\alpha}{\kappa}\right)+\frac{\alpha(2-\kappa)}{4}\ln\left(\frac{2-\kappa}{\alpha}\right)>0
\end{equation}
for all $\alpha,\kappa\in [0,2)$ and $\kappa\in[0,2-\alpha)$.
\end{prop}
\begin{proof}

To get the result we need to choose the paths for $z,w$ so that they will be steep descent. Let us consider the following paths:
\begin{equation}\label{eq56}
\begin{aligned}
\Gamma_{-1,\alpha-2}&=\left\{w=-1+\frac{\alpha}{2}+\I y-|y|,y\in\R\right\},\\
\Gamma_0&=\left\{z=-\frac{\kappa}{2}+\I y+|y|,y\in\R\right\}.
\end{aligned}
\end{equation}
With this choice, $\Gamma_0$ stays on the right of $\alpha-2-\Gamma_{-1,\alpha-2}$ since we assumed $\kappa<2-\alpha$, see Figure~\ref{FigPaths2}.
Now we verify the steep descent property of the paths. By symmetry it is enough to consider the portion of the paths in the upper-half plane.

\textit{Path $\Gamma_{-1,\alpha-2}$:} Consider $w=-1+\frac{\alpha}{2}+\I y-y$ for $y\geq 0$, $\tilde s\in [0,\alpha(2-\kappa)/4]$. Then,
\begin{equation}
\Re(f_0(w,\tilde s))=\mathrm{const} -y+\frac{\kappa(2-\alpha)}{8}\ln(|w|^2)- \frac12\left(\frac{\alpha(2-\kappa)}{4}-\tilde s\right)\ln(|w+1|^2),
\end{equation}
with $|w|^2=\frac{(2-\alpha)^2}{4}+(2-\alpha)y+2y^2$ and $|w+1|^2=\frac{\alpha^2}{4}-\alpha y+2y^2$. Thus,
\begin{equation}\label{eq61b}
\frac{\partial \Re(f_0(w,\tilde s))}{\partial y} = -1+
\frac{\kappa(2-\alpha)}{8|w|^2}\left(4y+2-\alpha\right)- \left(\frac{\alpha(2-\kappa)}{4}-\tilde s\right)\frac{4y-\alpha}{2|w+1|^2}.
\end{equation}
Now we consider two cases:\\[0.5em]
\textit{Case a: $0< y\leq \alpha/4$.} In this case,
\begin{equation}
\begin{aligned}
(\ref{eq61b})&\leq -1+
\frac{\kappa(2-\alpha)}{8|w|^2}\left(4y+2-\alpha\right)
- \frac{\alpha(2-\kappa)}{8}\frac{4y-\alpha}{|w+1|^2}\\
&=-y^2\frac{8 y^2+(4y+1-\alpha) (2-\alpha-\kappa)+2-\alpha}{2|w|^2|w+1|^2}<0
\end{aligned}
\end{equation}
for all $0<\alpha<2$ and $0\leq \kappa < 2-\alpha$.\\[0.5em]
\textit{Case b: $y\geq \alpha/4$.} In this case,
\begin{equation}
\begin{aligned}
(\ref{eq61b})&\leq -1+
\frac{\kappa(2-\alpha)}{8|w|^2}\left(4y+2-\alpha\right)\\
&=-\frac{(2-\kappa)\left(\frac{(2-\alpha)^2}{4}+(2-\alpha)y\right)+4y^2}{2|w|^2}<0
\end{aligned}
\end{equation}
for all $\kappa<2$.

Further, as $y\to\infty$, $\frac{\partial \Re(f_0(w,\tilde s))}{\partial y}\to -1$, i.e., $\Re(f_0(w,\tilde s))\simeq -y$. This implies that the estimates of the integrand in $w$ will have an exponential decay as $e^{-yt}$. Thus our chosen path $\Gamma_{-1,\alpha-2}$ is steep descent.

\textit{Path $\Gamma_0$:} Consider $z=-\frac{\kappa}{2}+\I y + y$ for $y\geq 0$. Then
\begin{equation}
\Re(-f_0(z,\tilde s))=\mathrm{const} -y-\frac{\kappa(2-\alpha)}{8}\ln(|z|^2)+ \frac12\left(\frac{\alpha(2-\kappa)}{4}-\tilde s\right)\ln(|z+1|^2),
\end{equation}
with $|z|^2=\frac{\kappa^2}{4}-\kappa y+2y^2$ and $|z+1|^2=\frac{(2-\kappa)^2}{4}+(2-\kappa) y+2y^2$. Thus, using $\tilde s\geq 0$,
\begin{equation}\label{eq63}
\begin{aligned}
\frac{\partial \Re(-f_0(z,\tilde s))}{\partial y} &= -1-
\frac{\kappa(2-\alpha)}{8|z|^2}\left(4y-\kappa\right)+ \left(\frac{\alpha(2-\kappa)}{4}-\tilde s\right)\frac{4y+2-\kappa}{2(|z+1|^2)}\\
&\leq-1-
\frac{\kappa(2-\alpha)}{8|z|^2}\left(4y-\kappa\right)+ \frac{\alpha(2-\kappa)}{8}\frac{4y+2-\kappa}{(|z+1|^2)}\\
&=-y^2\frac{8 y^2+(4 y+2-\kappa) (2-\alpha -\kappa)+\alpha\kappa}{2 |z|^2 |z+1|^2}<0
\end{aligned}
\end{equation}
for all $\kappa>0$, $y>0$, since we assumes $0<\alpha<2$ and $0\leq \kappa<2-\alpha<2$.

By these two results on the steep descent property, the exponential decay for large $y$, and the fact that $|z-w|$ remains bounded away from $0$, we get the bound
\begin{equation}
\begin{aligned}
\left|K^{(2)}_{t,{\rm resc}}(s_1,s_2)\right|&\leq \mathrm{const}\, t^{1/3} e^{t \Re(f_0((\alpha-2)/2,\tilde s_1))-t \Re(f_0(-\kappa/2,\tilde s_2))}\\
&= \mathrm{const}\, t^{1/3} e^{t [\frac{\alpha+\kappa-2}{2}+\frac{\kappa (2-\alpha)}{4}\ln(\frac{2-\alpha}{\kappa})-\frac{\alpha(2-\kappa)}{4}\ln(\frac{2-\kappa}{\alpha})]}e^{-s_2 \ln((2-\kappa)/\alpha) t^{1/3}}.
\end{aligned}
\end{equation}
Since $(2-\kappa)/\alpha>1$ and $s_2\geq -\ell_0$, the last term is at worse $e^{c \ell_0 t^{1/3}}$ with $c=\ln((2-\kappa)/\alpha)>0$. Further one can verify that  $F(\alpha,\kappa)>0$ for all \mbox{$\alpha\in [0,2)$} and $\kappa\in [0,2-\alpha)$. Thus $\mathrm{const}\, t^{1/3} e^{-t F(\alpha,\kappa)}e^{c \ell_0 t^{1/3}}\leq e^{-t F(\alpha,\kappa)/2}$ for $t$ large enough. We have obtained that
\begin{equation}\label{eq64a}
\left|K^{(2)}_{t,{\rm resc}}(s_1,s_2)\right|\leq e^{-t F(\alpha,\kappa)/2}
\end{equation}
for $t$ large enough.

By exactly the same argument, but using that $|z-(\alpha-2-w)|$ remains bounded away from zero, we can bound $K^{(1,a)}_{t,{\rm resc}}$, namely
\begin{equation}\label{eq64b}
\left|K^{(1,a)}_{t,{\rm resc}}(s_1,s_2)\right|\leq e^{-t F(\alpha,\kappa)/2}.
\end{equation}
\end{proof}

\begin{prop}[Convergence for $K^{(1,b)}_{t,{\rm resc}}$]\label{propConvHalfLinePt2}
For any $s_1,s_2$ in a bounded set,
\begin{equation}
\lim_{t\to\infty} K^{(1,b)}_{t,{\rm resc}}(s_1,s_2) = \sigma\Ai(\sigma(s_1+s_2))
\end{equation}
with $\sigma=\frac{(2-\alpha)^{2/3}}{(\alpha\left((2-\alpha)^2-2 (1-\alpha)\kappa\right))^{1/3}}$.
\end{prop}
\begin{proof}
We have
\begin{equation}
K^{(1,b)}_{t,{\rm resc}}(s_1,s_2) = \frac{t^{1/3}}{2\pi \I} \oint_{\Gamma_{-1,\alpha-2}}\frac{\mathrm{d}w}{w+1} e^{t [f_0(w,0)-f_0(\alpha-2-w,0)]} e^{t^{1/3} [s_1 f_2(w)-s_2 f_2(2-\alpha-w)]}
\end{equation}
with $f_2(w)=\ln(2(w+1)/\alpha)$.

First we show that $\Gamma_{-1,\alpha-2}$ as in (\ref{eq56}) is steep descent for
\begin{equation}\label{eq72}
g_0(w,\tilde s_1,\tilde s_2):=f_0(w,\tilde s_1)-f_0(\alpha-2-w,\tilde s_2),
\end{equation}
for $\tilde s_1,\tilde s_2\in [0,\alpha(2-\kappa)/4]$. It is a little bit more than what we need for this proposition, but we will use it in Proposition~\ref{propDecayHalfLinePt2b} again.
From the proof of Proposition~\ref{propDecayHalfLinePt2a} we already know that the path is steep descent for $f_0(w,\tilde s_1)$. Now consider $z=\alpha-2-w=-1+\frac{\alpha}{2}+\I y+y$, $y\geq 0$. Then, $|z|^2=\frac{(2-\alpha)^2}{4}-(2-\alpha)y+2y^2$ and $|z+1|^2=\frac{\alpha^2}{4}+\alpha y+2 y^2$. The same computation as in (\ref{eq63}) given, for $\tilde s\geq 0$,
\begin{equation}
\begin{aligned}
\frac{\partial \Re(-f_0(z,\tilde s))}{\partial y} &\leq-1-
\frac{\kappa(2-\alpha)}{8|z|^2}\left(4y-1+\alpha/2\right)+ \frac{\alpha(2-\kappa)}{8}\frac{4y+1+\alpha/2}{(|z+1|^2)}\\
&=-y^2\frac{8 y^2+(4y+1-\alpha) (2-\alpha-\kappa)+2-\alpha}{2 |z|^2 |z+1|^2}<0
\end{aligned}
\end{equation}
for all $y>0$ under our assumptions $0<\alpha<2$ and $0\leq \kappa < 2-\alpha$. Moreover, as $y\to\infty$, $\Re(-f_0(z,\tilde s))\simeq -y$. Putting together the two results, we have that the chosen path $\Gamma_{-1,\alpha-2}$ is steep descent for $g_0(w,\tilde s_1,\tilde s_2)$ and for $y\to\infty$ we have $\Re(g_0(w,\tilde s_1,\tilde s_2))\lesssim -2y$.

Therefore, the contribution to $K^{(1,b)}_{t,{\rm resc}}(s_1,s_2)$ coming from $|y|\geq \delta$ is of order $\Or(t^{1/3} e^{-c(\delta) t})$ for some $c(\delta)>0$. It remains to control the contribution for $|y|\leq \delta$. By Taylor series we have
\begin{equation}
g_0(w,0,0)=-Q(\alpha,\kappa) \frac{(2(\I-1)y/\alpha)^3}{3}+\Or(y^4),
\end{equation}
with
\begin{equation}\label{eqQ}
Q(\alpha,\kappa)=\frac{\alpha\left((2-\alpha)^2-2 (1-\alpha)\kappa\right)}{(2-\alpha)^2}
\end{equation}
and
\begin{equation}
s_1 f_2(w)-s_2 f_2(2-\alpha-w) = (s_1+s_2) 2(\I-1)y/\alpha +\Or(y^2).
\end{equation}
So, the contribution from $0\leq y\leq \delta$ is given by
\begin{equation}\label{eq71b}
\frac{t^{1/3}}{2\pi\I}\frac{2(\I-1)}{\alpha} \int_{0}^\delta \dx y  e^{-t Q(\alpha,\kappa) (2(\I-1)y/\alpha)^3/3+t^{1/3} (s_1+s_2) 2(\I-1)y/\alpha} e^{\Or(t y^4,t^{1/3} y^2)}.
\end{equation}
The cubic term has a prefactor with negative real part, so that it dominates all the error terms. Consider first (\ref{eq71b}) without the error terms. Then, by the change of variables $W:=-t^{1/3} Q(\alpha,\kappa)^{1/3} 2 (\I-1) y/\alpha$, we get
\begin{equation}
\frac{Q(\alpha,\kappa)^{-1/3}}{2\pi\I} \int_{-t^{1/3} Q(\alpha,\kappa)^{1/3} 2 (1-\I) \delta/\alpha}^0 \dx W e^{W^3/3-(s_1+s_2)Q(\alpha,\kappa)^{-1/3} W}.
\end{equation}
Extending the contour to $(\I-1)\infty$ the error term is only $\Or(e^{-c(\delta)t})$ and adding the contribution of $y\leq 0$ we finally get that the main contribution is given by
\begin{equation}
\frac{Q(\alpha,\kappa)^{-1/3}}{2\pi\I} \int_{-(1-\I)\infty}^{-(1+\I)\infty} \dx W e^{W^3/3-(s_1+s_2)Q(\alpha,\kappa)^{-1/3} W} = \sigma \Ai(\sigma(s_1+s_2))
\end{equation}
where we set $\sigma=Q(\alpha,\kappa)^{-1/3}$. Finally, to control the error terms in  (\ref{eq71b}), one uses as usual the identity $|e^{|x|}-1|\leq |x|e^{|x|}$ with $x$ replaced by the error terms, and obtains a contribution of order $\Or(t^{-1/3})$.
\end{proof}

\begin{prop}[Bounds for $K^{(1,b)}_{t,{\rm resc}}$]\label{propDecayHalfLinePt2b}
For any $\ell_0>0$, there exists a $t_0$ such that for $t>t_0$ and $s_{1},s_{2} \in [-\ell_0,\frac{\alpha(2-\kappa)}{4} t^{2/3}]$
\begin{equation}\label{eq78}
|K^{(1,b)}_{t,{\rm resc}}(s_1,s_2)| \leq  C e^{-(s_1+s_2)/2},
\end{equation}
for some finite constant $C$.
\end{prop}
\begin{proof}
The proof is very similar to the one in previous papers, see e.g.~Proposition~5.3 of~\cite{BF07}. We will skip some algebraic details and focus on the strategy and the key points. First, for any $t$-independent $\tilde\ell$ the result for $(s_1,s_2)\in[-\ell_0,\tilde \ell]^2$ follows from the proof of Proposition~\ref{propConvHalfLinePt2}. The constant $\tilde\ell$ can be chosen later and, for instance, if $(s_1,s_2)\in [-\ell_0,\infty)^2\setminus [-\ell_0,\tilde \ell]^2$, it can be chosen such that $s_1+s_2$ is large enough.

As before, we denote $\tilde s_i=s_i t^{-2/3}$. The integral we have to estimate is then
\begin{equation}\label{eq80}
\frac{t^{1/3}}{2\pi \I} \oint_{\Gamma_{-1,\alpha-2}}\frac{\mathrm{d}w}{w+1} e^{t g_0(w,\tilde s_1,\tilde s_2)}
\end{equation}
with $g_0$ given in (\ref{eq72}). We have seen in the first part of the proof of Proposition~\ref{propConvHalfLinePt2} that the path $\Gamma_{-1,\alpha-2}$ as in (\ref{eq56}) is steep descent for general values of $s_1,s_2$ in our domain. The idea is now to consider a minor modification of this path around  $w_c=-1+\alpha/2$ as follows, see Figure~\ref{FigPaths3}.
\begin{figure}
\begin{center}
\psfrag{w}[l][b]{$w$}
\psfrag{wc-rho}[l][b]{$w_c-\rho$}
\psfrag{0}[c][b]{$0$}
\psfrag{m1}[c][b]{$-1$}
\psfrag{wc}[c][b]{$w_c$}
\includegraphics[height=4cm]{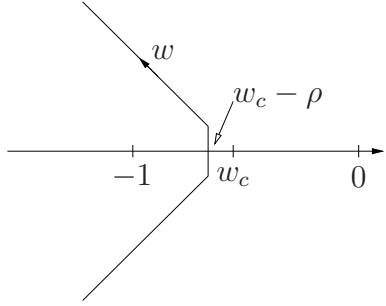}
\caption{Paths used for the asymptotic analysis in Proposition~\ref{propDecayHalfLinePt2b}.}
\label{FigPaths3}
\end{center}
\end{figure}

Consider
\begin{equation}\label{eq81}
w=w_c-\rho(1-\I y),\quad |y|\leq 1,
\end{equation}
where $\rho$ is chosen as follows:
\begin{equation}\label{eq82}
\begin{aligned}
\rho=\left\{
       \begin{array}{ll}
         \frac{\alpha}{2\sqrt{Q(\alpha,\kappa)}}\sqrt{\tilde s_1+\tilde s_2}, & \textrm{for }0\leq \tilde s_1+\tilde s_2\leq \e, \\
         \frac{\alpha}{2\sqrt{Q(\alpha,\kappa)}}\sqrt{\e}, & \textrm{for }\tilde s_1+\tilde s_2 \geq \e,
       \end{array}
     \right.
\end{aligned}
\end{equation}
with $Q=Q(\alpha,\kappa)$ given in (\ref{eqQ}). For the asymptotic analysis, $\e>0$ can be chosen as small as needed (but independent of $t$). This piece of contour joins the original path (\ref{eq56}). Now one has to control the real part of $g_0$ only in a neighborhood of $-1+\alpha/2$ (at a distance $\Or(\e)$ only). Taylor series at $w_c$ gives
\begin{equation}\label{eq83}
g_0(w,\tilde s_1,\tilde s_2)=-Q\frac{2^3}{\alpha^3}\frac{(w-w_c)^3}{3}+(\tilde s_1+\tilde s_2)\frac{2}{\alpha}(w-w_c) + \Or\left((w-w_c)^4,\tilde s_i (w-w_c)^2\right).
\end{equation}
For the choice in (\ref{eq81})-(\ref{eq82}), one looks for the minimal $w$ of (\ref{eq83}) without the error terms and gets the first choice. However, in order to have enough control through Taylor approximation, we have to stay in a small neighborhood of $w_c$. This is the reason for the $\e$ cut-off in (\ref{eq82}).

Replacing (\ref{eq81}) into the main part of (\ref{eq83}) one gets, for $0\leq \tilde s_1+\tilde s_2\leq \e$,
\begin{equation}\label{eq84}
\Re\left(-Q\frac{2^3}{\alpha^3}\frac{(w-w_c)^3}{3}+(\tilde s_1+\tilde s_2)\frac{2}{\alpha}(w-w_c)\right)=-\frac{(\tilde s_1+\tilde s_2)^{3/2} (2+3 y^2)}{3\sqrt{Q}},
\end{equation}
while for $\tilde s_1+\tilde s_2 \geq \e$,
\begin{equation}\label{eq85b}
\begin{aligned}
\Re\left(-Q\frac{2^3}{\alpha^3}\frac{(w-w_c)^3}{3}+(\tilde s_1+\tilde s_2)\frac{2}{\alpha}(w-w_c)\right)&=
-\frac{3(\tilde s_1+\tilde s_2)\sqrt{\e}+(3 y^2-1)\e^{3/2}}{3\sqrt{Q}}\\
&\leq -\frac{2(\tilde s_1+\tilde s_2)\sqrt{\e}+3 y^2\e^{3/2}}{3\sqrt{Q}}.
\end{aligned}
\end{equation}
The two key properties in (\ref{eq84}) and (\ref{eq85b}) are: (1) the quadratic decay of $e^{t g_0(w,\tilde s_1,\tilde s_2)}$ due the $y^2$ term, and (2) at $y=0$ one would have the bound
\begin{equation}
\begin{aligned}
e^{t\Re(g_0(w,\tilde s_1,\tilde s_2))} &\lesssim
\left\{
   \begin{array}{ll}
    e^{-\frac23 (s_1+s_2)^{3/2} Q^{-1/2}}, & \textrm{for }0\leq \tilde s_1+\tilde s_2\leq \e, \\
    e^{-\frac23 (s_1+s_2)\sqrt{\e}t^{1/3} Q^{-1/2}}, & \textrm{for }\tilde s_1+\tilde s_2 \geq \e,
   \end{array}
\right.
\end{aligned}
\end{equation}
by ignoring the error terms in (\ref{eq83}). For $s_1+s_2$ large enough and $t$ large enough, in both cases (\ref{eq85b}) is bounded by $e^{-c (s_1+s_2)}$ for any choice of $c>0$. By choosing $\e$ small enough, it is not so difficult (but a bit lengthy) to control the error terms in (\ref{eq83}) too. This can be made in exactly the same way as in the proof of Proposition~5.3 of~\cite{BF07} (see the argument between equations (5.40) and (5.47) in~\cite{BF07}). As a result, one obtains for instance a bound for the rescaled kernel (\ref{eq80}) like (\ref{eq85b}) with the prefactor $\frac23$ replaced by $\frac13$. This estimate is good enough and leads to the bound (\ref{eq78}).
\end{proof}

\begin{prop}\label{decdec}
Let $\eta>\frac{\alpha^2}{(2-\alpha)^2}$ and $\tilde{\mu}=2\big(\frac{\eta}{\alpha}+\frac{1}{2-\alpha}\big)$.
Then, for any  $\epsilon \geq 0$, there exist constants $C,\tilde{c}$ such that
\begin{equation}\label{Lmibound}
\Pb\left(L_{\mathcal{L}^{-}\rightarrow (\lfloor\eta\ell\rfloor,\lfloor \ell\rfloor)}> (\tilde{\mu}+\epsilon/2)\ell\right)
\leq C \exp(-\tilde{c}\epsilon \ell^{2/3}).
\end{equation}
\end{prop}
\begin{proof}
It is quite similar to the one of Proposition~\ref{propHalfLineBound}. We use again the correspondance \eqref{eqLPPtasep} between TASEP and LPP. We set
$t:=(\tilde{\mu}+\epsilon/2)\ell$, $\ell(t)=t/(\tilde{\mu}+\epsilon/2)$, Proposition~\ref{PropDistrLminus} tells us
\begin{equation}
\eqref{Lmibound}=1-\Pb(x_{\ell(t)}(t)\geq (\eta-1)\ell(t)).
\end{equation}
We denote
\begin{equation}
X_{t}^{\mathrm{resc}}=\frac{x_{\ell(t)}(t)-\frac{(\alpha-\kappa)t}{2}}{-t^{1/3}}
\end{equation}
with $\kappa=\frac{4}{2-\alpha}\big(\tilde{\mu}+\frac{\epsilon}{2}\big)^{-1}$ so that
$\ell(t)=\kappa \frac{2-\alpha}{4}t$.
Then,
\begin{equation}\label{Lmibound2}
\begin{aligned}
\eqref{Lmibound}&=1-\Pb(X_{t}^{\mathrm{resc}}\leq \frac{(\eta-1)\ell(t)-\frac{\alpha-\kappa}{2}t}{-t^{1/3}})
\\&=-\sum_{m=1}^{\infty}\frac{(-1)^{m}}{m!}\int \dx s_{1}\cdots \int \dx s_{m}\det[t^{1/3}\tilde{K}_{\ell(t),t}(x(s_i),x(s_j)]_{1\leq i,j\leq m},
\end{aligned}
\end{equation}
where $x(s)=\frac{\alpha-\kappa}{2}t-st^{1/3}$ and the integration domain of the $s_i$ is
\mbox{$(\alpha\epsilon t^{2/3}/4(\tilde{\mu}+\epsilon/2),\alpha  (2-\kappa) t^{2/3}/4]$}. This comes from the fact that
with $x=(\eta-1)\ell(t)$ we have
\begin{equation}
s =\frac{x-\frac{\alpha-\kappa}{2}t}{-t^{1/3}}=\frac{\alpha\epsilon t^{2/3}}{4(\tilde{\mu}+\epsilon/2)}
\end{equation}
together with the fact that the original kernel $K_{n,t}$ is identically equal to zero for $x(s)+n<0$.

A straightforward consequence of Proposition~\ref{propDecayHalfLinePt2a} is that, for \mbox{$s_1,s_2\in [-\ell_0,\alpha(2-\kappa)t^{2/3}/4]$} it holds
\begin{equation}
|K^{(1,a)}_{t,{\rm resc}}(s_1,s_2)|+|K^{(2)}_{t,{\rm resc}}(s_1,s_2)| \leq e^{-F(\alpha,\kappa)t/4} e^{-(s_1+s_2)/2}
\end{equation}
for $t$ large enough.
This together with the exponential bound of Proposition~\ref{propDecayHalfLinePt2b} implies that we can thus single out a factor $\prod_{i=1}^{m}C^{m}e^{-s_i}$ so that using Hadamard's bound, we get
\begin{equation} \begin{aligned}
|(\ref{Lmibound})|&\leq \sum_{m=1}^{\infty}\frac{C^{m}m^{m/2}}{m!}\int_{\e\alpha t^{2/3}/4(\tilde{\mu}+\e/2)}^{\alpha  (2-\kappa) t^{2/3}/4}\text{d}s_1\cdots\int_{\e \alpha t^{2/3}/4(\tilde{\mu}+\e/2)}^{\alpha  (2-\kappa) t^{2/3}/4} \text{d}s_{m}\prod_{i=1}^{m}e^{-s_i}
\\&\leq \tilde{C}\exp\left(-\tilde c \e \ell^{2/3}\right)
\end{aligned}\end{equation}
for some constants $\tilde C,\tilde c$ (uniform in $\ell$), where the last steps are identical to the ones of Proposition~\ref{propHalfLineBound}.
\end{proof}

\begin{prop}[Half-line $\mathcal{L}^-$-to-point LPP: convergence to $F_1$]\label{propHalfFlatConvergence2}
For any fixed $\eta>\frac{\alpha^2}{(2-\alpha)^2}$, it holds
\begin{equation} \begin{aligned}
\lim_{\ell\to\infty}\Pb\left(L_{\mathcal{L}^{-}\to (\lfloor\eta\ell\rfloor,\lfloor\ell\rfloor)}\leq \tilde{\mu}\ell +s \hat{ \sigma}_\eta \ell^{1/3}\right)= F_1(2s)
\end{aligned}\end{equation}
where $\tilde{\mu}=2(\frac{\eta}{\alpha}+\frac{1}{2-\alpha})$, $\hat{\sigma}_\eta=\frac{2^{4/3}}{\alpha}\left(\eta+\frac{\alpha^3}{(2-\alpha)^3}\right)^{1/3}$, and $F_1$ is the GOE Tracy-Widom distribution function.
\end{prop}
\begin{proof}
First, with $\sigma$ as in Proposition~\ref{propConvHalfLinePt2}, it holds
\begin{equation}\label{eq85}
\begin{aligned}
\Pb\left(L_{\mathcal{L}^{-}\to (\lfloor\eta\ell\rfloor,\lfloor\ell\rfloor)}\leq \tilde{\mu}\ell +s \hat{ \sigma}_\eta \ell^{1/3}\right)
&= \Pb\left(x_{\ell}(\tilde\mu \ell + s \hat{ \sigma}_\eta \ell^{1/3})\geq (\eta-1)\ell\right)\\
&= \Pb\left(x_{[\kappa(2-\alpha)t/4]}(t)\geq \frac{\alpha-\kappa}{2}t-\sigma^{-1} s t^{1/3}\right)
\end{aligned}
\end{equation}
if we choose
\begin{equation}
\begin{aligned}
t = \tilde \mu +s \hat{ \sigma}_\eta \ell^{1/3} &\Leftrightarrow \ell=\frac{t}{\tilde\mu}-\frac{s\hat\sigma_\eta t^{1/3}}{\tilde\mu^{4/3}}+o(1),\\
\frac{\kappa(2-\alpha)}{4}t=\ell &\Leftrightarrow \kappa=\frac{4}{2-\alpha}\left(\frac{1}{\tilde\mu}-\frac{s\hat\sigma_\eta t^{-2/3}}{\tilde\mu^{4/3}}\right),
\end{aligned}
\end{equation}
and finally $\frac{\alpha-\kappa}{2}t-\sigma^{-1} s t^{1/3} =(\eta-1)\ell$, which fixes the values of $\tilde\mu$ and $\hat\sigma_\eta$ as given in the statement. Now, the r.h.s.~of (\ref{eq85}) is given by a Fredholm determinant like in (\ref{Lmibound2}), with the minor difference that now the lower integration bound is simply given by $s$ and that the scaling of the kernel has the extra $\sigma^{-1}$ in front. From Propositions~\ref{propDecayHalfLinePt2a} and~\ref{propDecayHalfLinePt2b} we know that the kernel is uniformly bounded (in $t$) by a function so that its Fredholm series is bounded. Thus we can apply dominated convergence to take the limit inside the Fredholm series. Finally, Proposition~\ref{propConvHalfLinePt2} tells us that the pointwise limit of the rescaled kernel (including the extra $\sigma^{-1}$ factor in the spatial scaling) converges pointwise to $\Ai(s_1+s_2)$. Thus,
\begin{equation}
\lim_{t\to\infty} \Pb\left(x_{[\kappa(2-\alpha)t/4]}(t)\geq \frac{\alpha-\kappa}{2}t-\sigma^{-1} s t^{1/3}\right)=F_1(2s),
\end{equation}
which ends the proof.
\end{proof}

\begin{cor}\label{CorLminus}
Fix an $\eta>\alpha^2/(2-\alpha)^2$, a $\beta\in (1/3,1]$ and define
\begin{equation}
\gamma \in [0,1-t^{\beta-1}],\quad \e  = t^{-\chi}\textrm{ with } \chi\in (0,2/3).
\end{equation}
Then there exists constants $C,\tilde{c}>0$ and $t_0 >0$ such that for all $t >t_0$
\begin{equation}
\begin{aligned}
\Pb\left(L_{\mathcal{L}^{-}\to D_\gamma}> \left(\tilde\mu_{\gamma}+\frac{\epsilon }{2}\right)t\right) \leq C\exp\left(-\tilde{c} \,t^{2/3-\chi}\right),
\end{aligned}
\end{equation}
where $\tilde \mu_{\gamma}= 2 \gamma \left(\frac{\eta}{\alpha}+\frac{1}{2-\alpha}\right) $.
\end{cor}
\begin{proof}
It is a straightforward consequence of Proposition~\ref{decdec}. Indeed, with $\ell=\gamma t$, we have
\begin{equation}
\begin{aligned}
\Pb\left(L_{\mathcal{L}^{-}\to D_\gamma}> \left(\tilde{\mu}_{\gamma}+\epsilon/2\right)t\right)
&=\Pb\left(L_{\mathcal{L}^{-}\to (\lfloor\eta\ell\rfloor,\lfloor\ell\rfloor)}> \left(\tilde{\mu}+\epsilon/(2\gamma)\right)\ell\right)
\\&\leq\Pb\left(L_{\mathcal{L}^{-}\to (\lfloor\eta\ell\rfloor,\lfloor\ell\rfloor)}> \left(\tilde{\mu}+\epsilon/2\right)\ell\right)
\\&\leq C\exp(-\tilde{c}t^{2/3-\chi}),
\end{aligned}
\end{equation}
where the second inequality holds since $\gamma\leq 1$.
\end{proof}
\subsection{No-crossing results}\label{sectProofNocrossing}
In this section we collect the non-crossing results, which are proven below.

\begin{prop}\label{nocrossingPlus}
Consider the point \mbox{$E=(\lfloor \eta  t \rfloor,\lfloor  t\rfloor)$} for $0<\eta< 1$ (see Figure~\ref{FigLplusminus}). For some fixed $\beta\in (1/3,1]$, consider the points $D_{\gamma}=(\lfloor \gamma \eta t \rfloor,\lfloor  \gamma t\rfloor)$ with $\gamma \in [0,1-t^{\beta-1}]$. Then, for all $t$ large enough
\begin{equation} \begin{aligned}
\Pb\bigg(\bigcup_{D_{\gamma}\atop \gamma \in [0,1-t^{\beta-1}]}\{D_\gamma\in \pi^{\rm max}_{\mathcal{L}^{+}\to E}\}\bigg)\leq C \exp(-c t^{\beta-1/3}),
\end{aligned}
\end{equation}
for some $t$-independent constants $C,c>0$.
\end{prop}

\begin{figure}[h!]
\begin{center}
\psfrag{A}[l][b]{$E=(\eta t,t)$}
\psfrag{AA}[r][b]{$E^+$}
\psfrag{D}[l][b]{$D_\gamma=(\gamma \eta t,\gamma t)$}
\psfrag{B}[l][b]{$B$}
\psfrag{Q}[c][b]{$((\eta-1)t,0)$}
\psfrag{Zp}[r][b]{$Z^+$}
\psfrag{Zm}[r][b]{$Z^-$}
\psfrag{Lp}[l][b]{$\mathcal{L}^+$}
\psfrag{Lm}[l][b]{$\mathcal{L}^-$}
\includegraphics[height=7cm]{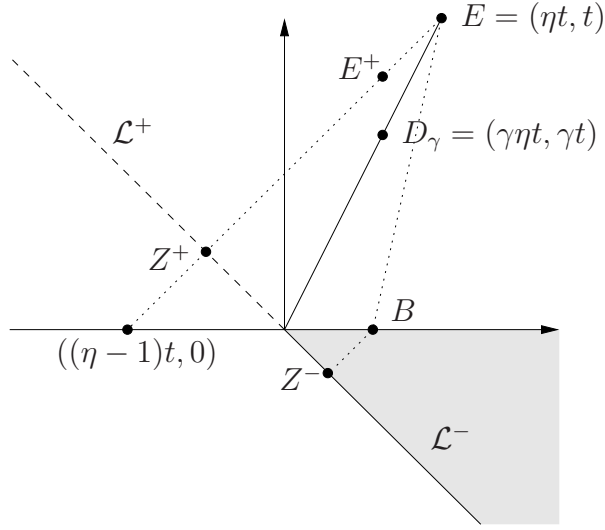}
\caption{Illustration of the geometry for the LPP of \mbox{Propositions~\ref{nocrossingPlus}--\ref{nocrossingPtPt}}. The half-line $\mathcal{L}^-$ is the solid one, while the half-line $\mathcal{L}^+$ is the dashed one. Further, $E^{+}=(\eta  t -t^{\nu},t-t^{\nu})$, $B=(\eta-\alpha^2/(2-\alpha)^2)(t,0)$, \mbox{$Z^+=(1-\eta)(-t/2,t/2)$}, and \mbox{$Z^-=(\eta-\alpha^2/(2-\alpha)^2)(t/2,-t/2)$}. In the grey regions, the exponential random variables have parameter $\alpha\in (0,2)$, while in the white regions, they have parameter $1$. }
\label{FigLplusminus}
\end{center}
\end{figure}
\begin{prop}\label{nocrossing+}
Consider the point \mbox{$E^{+}=(\lfloor \eta  t -t^{\nu}\rfloor,\lfloor  t-t^{\nu}\rfloor)$} for \mbox{$0<\eta<1$} and $1/3<\nu<1$(see Figure~\ref{FigLplusminus}). For some fixed $\beta\in (1/3,1]$, consider the points $D_{\gamma}=(\lfloor \gamma \eta t) \rfloor,\lfloor  \gamma t\rfloor)$ with $\gamma \in [0,1-t^{\beta-1}]$. Then, for all $t$ large enough
\begin{equation}
\begin{aligned}
\Pb\bigg(\bigcup_{D_{\gamma}\atop \gamma \in [0,1-t^{\beta-1}]}\{D_\gamma\in \pi^{\rm max}_{\mathcal{L}^{+}\to E^{+}}\}\bigg)\leq C \exp(-c t^{\beta-1/3}),
\end{aligned}
\end{equation}
for some $t$-independent constants $C,c>0$.
\end{prop}

\begin{prop}\label{nocrossingMinus} Let $\alpha \in (0,2)$.
Consider, for some  $\eta>\alpha^2/(2-\alpha)^2$, the point \mbox{$E=(\lfloor \eta  t \rfloor,\lfloor  t\rfloor)$} as in Figure~\ref{FigLplusminus}. For some fixed $\beta\in (1/3,1]$, consider the points \mbox{$D_{\gamma}=(\lfloor \gamma \eta t \rfloor,\lfloor  \gamma t\rfloor)$} with $\gamma \in [0,1-t^{\beta-1}]$. Then, for all $t$ large enough
\begin{equation}
\begin{aligned}
\Pb\bigg(\bigcup_{D_{\gamma}\atop \gamma \in [0,1-t^{\beta-1}]}\{D_\gamma\in \pi^{\rm max}_{\mathcal{L}^{-}\to E}\}\bigg)\leq C \exp(-c t^{\beta-1/3}),
\end{aligned}
\end{equation}
for some $t$-independent constants $C,c>0$.
\end{prop}

Similarly, for the point-to-point geometry we have:
\begin{prop}\label{nocrossingPtPt}
Consider the point \mbox{$E=(\lfloor \eta  t \rfloor,\lfloor  t\rfloor)$} for $0<\eta< 1$. For some fixed $\beta\in (1/3,1]$, consider the points $D_{\gamma}=(\lfloor \gamma \eta t \rfloor,\lfloor  \gamma t\rfloor)$ with \mbox{$\gamma \in [0,1-t^{\beta-1}]$}. Then, for all $t$ large enough
\begin{equation} \begin{aligned}
\Pb\bigg(\bigcup_{D_{\gamma}\atop \gamma \in [0,1-t^{\beta-1}]}\{D_\gamma\in \pi^{\rm max}_{(\lfloor(\eta-1)t\rfloor,0)\to E}\}\bigg)\leq C \exp(-c t^{\beta-1/3}),
\end{aligned}
\end{equation}
for some $t$-independent constants $C,c>0$.
\end{prop}

\begin{prop}\label{nocrossingsymmetric}
  For some fixed $\beta\in (1/3,1]$, consider the points \mbox{$D_{\gamma}=(\lfloor \gamma  t \rfloor,\lfloor  \gamma t\rfloor)$} with $\gamma \in [0,1-t^{\beta-1}]$. Then, for all $t$ large enough
\begin{equation}
\begin{aligned}
\Pb\bigg(\bigcup_{D_{\gamma}\atop \gamma \in [0,1-t^{\beta-1}]}\{D_\gamma\in \pi^{\rm max}_{(-t,0)\to (t,t)}\}\bigg)\leq C \exp(-c t^{\beta-1/3}),
\end{aligned}
\end{equation}
for some $t$-independent constants $C,c>0$.
\end{prop}

\subsubsection{Proof of Propositions~\ref{nocrossingPlus},~\ref{nocrossing+},~\ref{nocrossingPtPt},~\ref{nocrossingsymmetric}}\label{sectProofNocrossingLplus}
In order to prove Proposition~\ref{nocrossingPlus}, we will adopt the notation and line of argumentation of a proof due to Johansson, namely the Lemmas~3.1, 3.2 and~3.3 in~\cite{Jo00}. He used them to prove that the maximizing path of a LPP model in Poisson points does not leave  a cylinder of width $N^{2/3}$ as $N\to\infty$.

Using the deviation results from the previous section, we first show that the probability that for some $\gamma$ the LPP-times $L_{\mathcal{L}^{+} \to D_\gamma}$ and  $L_{D_\gamma\to E}$ exceed by $\e  t/2$ their leading orders converges to zero.
\begin{prop} \label{PropLplusA}
Fix an $0< \eta < 1$, a $\beta\in (1/3,1]$, a $\chi\in (0,2/3)$. Let us set $\e = t^{-\chi}$.
We define a finite\footnote{The family is finite even if $\gamma$ is uncountable, since the number of different $D_\gamma$ is finite.} family of events $\left\{E_{D_{\gamma}}\right\}_{\gamma \in [0,1-t^{\beta-1}]}$ via
\begin{equation}
\begin{aligned}\label{bbb}
E_{D_\gamma}:=&\{\omega: L_{\mathcal{L}^{+}\to D_\gamma}(\omega)\leq (\mu_\gamma +\e /2)t\}
\cap \{L_{ D_\gamma\to E}(\omega)\leq (\mu_{{\rm pp},\gamma} +\e /2)t\},
\end{aligned}
\end{equation}
where
\begin{equation}
\mu_{\gamma}=2(1+\eta)\gamma,\quad \mu_{{\rm pp},\gamma}=(1-\gamma)(1+\sqrt{\eta})^{2}.
\end{equation}
Then
\begin{equation}\label{eq29}
\Pb\bigg(\bigcup_{D_\gamma} \Omega\setminus E_{D_\gamma}\bigg)\leq C' \exp(-c' t^{2/3-\chi})
\end{equation}
for some constants $C',c'>0$.
\end{prop}
\begin{proof}
To get the result, notice that there are $\mathcal{O}(t)$ many points $D_\gamma$, \mbox{$\gamma \in [0,1-t^{\beta-1}]$}, so that it is enough to get a good bound (uniform in $\gamma$) of $\Pb\left(\Omega\setminus E_{D_\gamma}\right)$. We have
\begin{equation}\label{devoneb}
\Pb\left(\Omega\setminus E_{D_\gamma}\right)\leq \Pb(L_{\mathcal{L}^{+} \to D_\gamma}\geq (\mu_\gamma +\e /2)t) + \Pb( L_{D_\gamma\to E}\geq (\mu_{{\rm pp},\gamma} +\e /2)t).
\end{equation}
According to  Corollary~\ref{CorLplus} there is a $t_0$ such that for $t>t_0$ we get
\begin{equation}\label{devonefirst}
\Pb(L_{\mathcal{L}^{+} \to D_\gamma}\geq (\mu_\gamma +\e /2)t)\leq C \exp(-\tilde{c} t^{2/3-\chi}).
\end{equation}
Remark that (with $\overset{d}{=}$ designating equality in distribution)
\begin{equation}\label{eqd}
L_{D_\gamma\to E} \overset{d}{=}  L_{0\to(\lfloor (1-\gamma) \eta t\rfloor,\lfloor(1-\gamma) t\rfloor)}.
\end{equation}
Furthermore, Proposition~\ref{devone} with $\ell=(1-\gamma) t$ and
$s=\frac{\e  t^{2/3}}{(1- \gamma)^{1/3}}$ gives
 \begin{equation}\label{devonetwice}
\Pb( L_{D_\gamma\to E}\geq (\mu_{{\rm pp},\gamma} +\e /2)t)\leq C \exp\left(-\e t^{2/3}\frac{c}{(1- \gamma)^{1/3}}\right)\leq C \exp(-c t^{2/3-\chi}).
\end{equation}
The bounds (\ref{devonefirst}) and (\ref{devonetwice}) imply that, for some constants $C',c'$,
\begin{equation}
\Pb\left(\Omega\setminus E_{D_\gamma}\right)\leq C' \exp(- c' t^{2/3-\chi}).
\end{equation}
Being the number of $D_\gamma$ of order $t$ only, the claimed bound holds true.
\end{proof}

Now we know that if a path goes through a point $D_\gamma$, then its typical last passage time is smaller than $(\mu_\gamma+\mu_{{\rm pp},\gamma}+2\e)t$. However, the typical last passage time of the maximizing paths is $\mu t,$ which is much larger.
\begin{prop}\label{PropLplusB}
Fix an $0< \eta < 1$, a $\beta\in (1/3,1]$, and $\gamma\in [0,1-t^{\beta-1}]$. Let us set \mbox{$\e = C t^{\beta-1}$}.
Then for all $t>0$ it holds
\begin{equation}
\frac{(\mu_{\gamma}+\mu_{{\rm pp},\gamma}+\e -\mu)t}{t^{1/3}}\leq -Ct^{\beta-1/3},
\end{equation}
with $C= (1-\sqrt{\eta})^2/2$, and
\begin{equation}
\mu=2(1+\eta),\quad \mu_{\gamma}=2(1+\eta)\gamma,\quad \mu_{{\rm pp},\gamma}=(1-\gamma)(1+\sqrt{\eta})^{2}.
\end{equation}
\end{prop}
\begin{proof}
A simple computations gives, for $0<\eta<1$,
\begin{equation}
\begin{aligned}
\frac{(\mu_{\gamma}+\mu_{{\rm pp},\gamma}+\e -\mu)t}{t^{1/3}}&=t^{\beta-1/3} (1-\sqrt{\eta})^2/2- (1-\gamma)(1-\sqrt{\eta})^2\\
&\leq -t^{\beta-1/3} (1-\sqrt{\eta})^2/2,
\end{aligned}
\end{equation}
where we used $1-\gamma\geq t^{\beta-1}$.
\end{proof}

We can now proceed to the final Proposition.
\begin{prop}\label{PropLplusC}
Fix an $0< \eta < 1$, a $\beta\in (1/3,1]$ and $\gamma\in [0,1-t^{\beta-1}]$.
Then, there exists a $t_0>0$ such that for all $t\geq t_0$ it holds
\begin{equation}\label{mouad}
\Pb(\{\omega:D_\gamma\in \pi^{\rm max}_{\mathcal{L}^{+}\to E}(\omega)\}) \leq C\exp(-c\, t^{\beta-1/3}),
\end{equation}
for some $t$-independent constants $C,c>0$.
\end{prop}
\begin{proof}[Proof of Proposition~\ref{PropLplusC}]
Denote by $I_{D_\gamma}$ the event that the maximizers from $\mathcal{L}^+$ to $E$ passes by the point $D_\gamma$, namely
\begin{equation}
I_{D_\gamma}=\{\omega:D_\gamma\in \pi^{\rm max}_{\mathcal{L}^{+}\to E}(\omega)\}.
\end{equation}
Let us choose $\e = t^{\beta-1} (1-\sqrt{\eta})^2/2$. Then,
\begin{equation}
\Pb(I_{D_\gamma})\leq \Pb\bigg(I_{D_\gamma}\cap \Big(\bigcap_{D_\gamma} E_{D_\gamma}\Big)\bigg)+\Pb\bigg(\Big(\bigcap_{D_\gamma} E_{D_\gamma}\Big)^c\bigg).
\end{equation}
The second term is exactly (\ref{eq29}) with $\chi=1-\beta$ (the extra coefficient in the definition of $\e$ is irrelevant, since it just modifies the value of the constant $c'$). Thus, the decay of the second term is as $\exp(-c' t^{\beta-1/3})$.

To bound the first term, notice that if $\omega \in I_{D_\gamma}$ and at the same time in each of the $E_{D_\gamma}$'s, then by Propositions~\ref{PropLplusA} and~\ref{PropLplusB},
\begin{equation}
\begin{aligned}
L_{\mathcal{L}^+\to E}(\omega) &\leq (\mu_\gamma+\mu_{{\rm pp},\gamma}+\e)t=\mu t + (\mu_\gamma+\mu_{{\rm pp},\gamma}+\e-\mu)t \\
&\leq \mu t - (C t^{\beta-1/3}) t^{1/3}.
\end{aligned}
\end{equation}
Therefore,
\begin{equation}\label{eq41}
\Pb\bigg(I_{D_\gamma}\cap \Big(\bigcap_{D_\gamma} E_{D_\gamma}\Big)\bigg) \leq
\Pb(L_{\mathcal{L}^+\to E}\leq \mu t - (C t^{\beta-1/3}) t^{1/3}).
\end{equation}
Further, denote by $Z^+$ the orthogonal projection of $E$ on $\mathcal{L}^{+}$, i.e., \mbox{$Z^+=\lfloor \frac{1-\eta}{2} \rfloor (-1,1)$}. Then, since $L_{\mathcal{L}^{+}\to E}\geq L_{Z^+ \to E}$, it follows that
\begin{equation}\label{eq42}
(\ref{eq41})\leq \Pb(L_{Z^+\to E}\leq \mu t - (C t^{\beta-1/3}) t^{1/3}).
\end{equation}
Moreover, since $L_{Z^+\to E}\overset{d}{=}L_{0\to \left(\lfloor \frac{1+\eta}{2}t\rfloor, \lfloor \frac{1+\eta}{2}t\rfloor \right)}$ we can apply the bound of Proposition~\ref{devtwo} (with $\ell \to (1+\eta) t/2$, $\eta \to 1$, and $s\ell^{1/3}\to C t^\beta$) to obtain
\begin{equation}\label{eq43bb}
(\ref{eq42})\leq \tilde C \exp(-\tilde c t^{3\beta/2-1/2})
\end{equation}
for some constants $\tilde C,\tilde c>0$.

Since for $\beta\in (1/3,1]$ and $\beta-1/3 \leq 3\beta/2-1/2$, then for all $t$ large enough
\begin{equation}
\Pb(I_{D_\gamma})\leq C \exp(-c t^{\beta-1/3}),
\end{equation}
for some $t$-independent constants $C,c>0$, which is the claimed result.
\end{proof}

\begin{proof}[Proof of Proposition~\ref{nocrossingPlus}]
The proof is a straightforward consequence of Proposition~\ref{PropLplusC}, since the cardinality of the family of points $\{D_\gamma\}_{\gamma\in [0,1-t^{\beta-1}]}$ is only of order $t$.
\end{proof}

\begin{proof}[Proof of Proposition~\ref{nocrossing+}]
The proof is very similar to the one of Proposition~\ref{nocrossingPlus}. Note first that for
$\gamma> 1-t^{\nu-1}$ then $\Pb(D_\gamma \in \pi_{\mathcal{L}^{+}\to E^{+}}^{\max})=0$.
For the analogue of Proposition~\ref{PropLplusA}, one only has to replace $E$ by $E^{+}$ in \eqref{bbb},
which amounts to replace $\eta$ by  $\tilde{\eta}=\frac{(1-\gamma)\eta t-t^{\nu}}{(1-\gamma)t-t^{\nu}}\to_{t\to\infty} \eta$  in \eqref{eqd},
$\mu_{\mathrm{pp},\gamma}$ by $\mu_{\mathrm{pp},\gamma}^{+}=(1-\gamma-t^{\nu-1})\left(1+\sqrt{\frac{\eta-t^{\nu-1}}{1-t^{\nu-1}}}\right)^{2}\to_{t\to\infty} \mu_{\mathrm{pp},\gamma}$ and apply Proposition~\ref{devone} to this new point-to-point LPP.
 The following  analogue
 of Proposition~\ref{PropLplusB} is a bit different.
\begin{prop}\label{PropLplusB+}
Fix an $0< \eta < 1$, a $\nu,\beta\in (1/3,1)$ , and $\gamma\in [0,1-t^{\beta-1}]$. Let us set \mbox{$\e = C t^{\beta-1}$}.
Then for all $t$ large it holds
\begin{equation}
\frac{(\mu_{\gamma}^{+}+\mu_{{\rm pp},\gamma}^{+}+\e -\mu^{+})t}{t^{1/3}}\leq -C t^{\beta-1/3},
\end{equation}
with $C= (1-\sqrt{\eta})^2/4$, and
\begin{equation}
\mu^{+}=2(1+\eta)-4t^{\nu-1},\quad \mu_{\gamma}^{+}=\gamma\mu^+,\quad \mu_{{\rm pp},\gamma}^{+}=(1-\gamma-t^{\nu-1})\left(1+\sqrt{\frac{\eta-t^{\nu-1}}{1-t^{\nu-1}}}\right)^{2}.
\end{equation}
\end{prop}
\begin{proof}
Using $\sqrt{\frac{\eta-t^{\nu-1}}{1-t^{\nu-1}}}<\sqrt{\eta}$ for $\eta<1$, we have $\mu_{{\rm pp},\gamma}^{+}\leq (1-\gamma)(1+\sqrt{\eta})^2$ so that \begin{equation}\label{eq123}
\frac{(\mu_{\gamma}^{+}+\mu_{{\rm pp},\gamma}^{+}+\e -\mu^{+})t}{t^{1/3}}\leq C t^{\beta-1/3}-(1-\gamma)\left(t^{2/3}(1-\sqrt{\eta})^2-4t^{\nu-1/3}\right).
\end{equation}
Then, using $\nu<1$ and $1-\gamma\geq t^{\beta-1}$ we have, for $t$ large enough,
\begin{equation}
(\ref{eq123}) \leq C t^{\beta-1/3}-t^{\beta-1/3}(1-\sqrt{\eta})^2/2=-C t^{\beta-1/3}.
\end{equation}
\end{proof}
With these two analogous statements at hand, we can adopt the proof of Proposition~\ref{PropLplusC},
simply replace  again $E$ by $E^{+}$ in \eqref{eq42}, and then again apply Proposition~\ref{devtwo} with $\ell\to \frac{1+\eta}{2}t-t^{\nu}$ to obtain a bound analogous to \eqref{eq43bb},  which finishes the proof.
\end{proof}

\begin{proof}[Proof of Proposition~\ref{nocrossingPtPt}]
The proof of Proposition~\ref{nocrossingPtPt} is almost identical, so let us indicate just the minor modifications.
What we have to do is to replace $\mathcal{L}^+$ with the point $(\lfloor(\eta-1)t\rfloor,0)$, now $\mu=4$ and $\mu_\gamma=4\gamma$. Further, there is one simplification, namely, the step (\ref{eq42}) is not needed (we would have equality in there).
\end{proof}

\begin{proof}[Proof of Proposition~\ref{nocrossingsymmetric}]
The analogue of  Proposition~\ref{PropLplusA} can be proven almost identically, one has
$\mu_\gamma=\left(1+\sqrt{\frac{1+\gamma}{\gamma}}\right)^{2}\gamma, \mu_{{\rm pp},\gamma}=4(1-\gamma)$ and uses twice Proposition~\ref{devone}.

The analogue of Proposition~\ref{PropLplusB} is again a bit different.
\begin{prop}\label{PropLplusB++}
Fix a $\beta\in (1/3,1]$ , and $\gamma\in [0,1-t^{\beta-1}]$. Let us set \mbox{$\e = C t^{\beta-1}$}.
Then for all $t$ large it holds
\begin{equation}
\frac{(\mu_{\gamma}+\mu_{{\rm pp},\gamma}+\e -\mu)t}{t^{1/3}}\leq -C t^{\beta-1/3},
\end{equation}
with $C= (3-2\sqrt{2})/4$, and
\begin{equation}
\mu_\gamma=\left(1+\sqrt{\frac{1+\gamma}{\gamma}}\right)^{2}\gamma,\quad \mu=(1+\sqrt{2})^{2},\quad \mu_{{\rm pp},\gamma}=4(1-\gamma).
\end{equation}
\end{prop}
\begin{proof}[Proof of Proposition~\ref{PropLplusB++}]
We have
\begin{equation}
\mu_{\gamma}+\mu_{{\rm pp},\gamma}-\mu=2 \left(\sqrt{\frac{1}{\gamma }+1}-1\right) \gamma -2 \sqrt{2}+2
\end{equation}
that is increasing in $\gamma$. Further, it holds
\begin{equation}
\mu_{\gamma}+\mu_{{\rm pp},\gamma}-\mu=(\gamma-1)\frac{3-2\sqrt{2}}{\sqrt{2}}+\Or((\gamma-1)^2).
\end{equation}
Thus by choosing $\gamma=1-t^{\beta-1}$ we get
\begin{equation}
\begin{aligned}\label{calc}
\frac{(\mu_{\gamma}+\mu_{{\rm pp},\gamma}+\e -\mu)t}{t^{1/3}}&\leq -t^{\beta-1/3} \left(\frac{3-2\sqrt{2}}{\sqrt{2}}-C\right)+  \Or(t^{2(\beta-1/3)})\leq -C t^{\beta-1/3}
\end{aligned}
\end{equation}
for $t$ large enough.
\end{proof}
The analogue of Proposition~\ref{PropLplusC} can be proven almost identically, the only difference being that the step~\eqref{eq42} is not needed.
\end{proof}

\subsubsection{Proof of Proposition~\ref{nocrossingMinus}}\label{sectProofNocrossingLminus}
The proof is very close to the one of Proposition~\ref{nocrossingPlus}, therefore we will skip some of the details, focusing more on the differences.
\begin{prop}\label{PropLminusA}
Fix an $\eta>\alpha^2/(2-\alpha)^2$, a $\beta\in (1/3,1]$ and  a $\chi\in (0,2/3)$. Let us set $\e = t^{-\chi}$.
We define a finite family of events $\big\{\tilde E_{D_{\gamma}}\big\}_{\gamma \in [0,1-t^{\beta-1}]}$ via
\begin{equation}
\begin{aligned}
\tilde E_{D_\gamma}:=&\{\omega: L_{\mathcal{L}^{-}\to D_\gamma}(\omega)\leq (\tilde \mu_\gamma +\e /2)t\}
\cap \{L_{ D_\gamma\to E}(\omega)\leq (\mu_{{\rm pp},\gamma} +\e /2)t\},
\end{aligned}
\end{equation}
where
\begin{equation}\label{eq66}
\tilde \mu_{\gamma}= 2 \gamma \left(\frac{\eta}{\alpha}+\frac{1}{2-\alpha}\right),\quad \mu_{{\rm pp},\gamma}=(1-\gamma)(1+\sqrt{\eta})^{2}.
\end{equation}
Then
\begin{equation}
\Pb\bigg(\bigcup_{D_\gamma} \Omega\setminus \tilde E_{D_\gamma}\bigg)\leq C' \exp(-c' t^{2/3-\chi})
\end{equation}
for some constants $C',c'>0$.
\end{prop}
\begin{proof}
The proof is like the one of Proposition~\ref{PropLplusA}, with the only difference that we employ Corollary~\ref{CorLminus} instead of Corollary~\ref{CorLplus} to control the decay of $\Pb(L_{\mathcal{L}^{-} \to D_\gamma}\geq (\tilde{\mu}_\gamma +\e /2)t)$.
\end{proof}

Now we know that if a path goes through the a point $D_\gamma$, then its typical last passage time is smaller than $(\tilde \mu_\gamma+\mu_{{\rm pp},\gamma}+2\e)t$. However, the typical last passage time of the maximizing path is $\tilde{\mu} t$ which is much larger.
\begin{prop}\label{PropLminusB}
Fix $\eta>\alpha^2/(2-\alpha)^2$, $\beta\in (1/3,1]$, and $\gamma\in [0,1-t^{\beta-1}]$.
Let us set  \mbox{$\e = C t^{\beta-1}$}. Then for all $t>0$ it holds
\begin{equation}\label{eq68}
\frac{(\tilde \mu_{\gamma}+\mu_{{\rm pp},\gamma}+\e -\tilde\mu)t}{t^{1/3}}\leq -\tilde{C}t^{\beta-1/3},
\end{equation}
with $C=\frac{(\alpha-(2-\alpha)\sqrt{\eta})^2}{2\alpha(2-\alpha)}$, and
\begin{equation}
\tilde\mu= 2 \left(\frac{\eta}{\alpha}+\frac{1}{2-\alpha}\right),\quad \tilde \mu_{\gamma}= \gamma \tilde\mu, \quad \mu_{{\rm pp},\gamma}=(1-\gamma)(1+\sqrt{\eta})^{2}.
\end{equation}
\end{prop}
\begin{proof}
A simple computations gives,
\begin{equation}
\begin{aligned}
\frac{(\tilde \mu_{\gamma}+\mu_{{\rm pp},\gamma}+\e -\tilde\mu)t}{t^{1/3}} &= (\gamma-1) \frac{(\alpha-(2-\alpha)\sqrt{\eta})^2}{\alpha(2-\alpha)}t^{2/3}+C t^{\beta-1/3}\\
&\leq -t^{\beta-1/3} \left( \frac{(\alpha-(2-\alpha)\sqrt{\eta})^2}{\alpha(2-\alpha)}-C\right) \leq -C t^{\beta-1/3}
\end{aligned}
\end{equation}
where we used $\alpha<1$ and $\gamma-1\leq -t^{\beta-1}$ and the fact that $\eta>\alpha^2/(2-\alpha)^2$.
\end{proof}

We can now proceed to the final proposition.
\begin{prop}\label{PropLminusC}
Fix an $\eta>\alpha^2/(2-\alpha)^2$, a $\beta\in (1/3,1]$ and let \mbox{$\gamma\in [0,1-t^{\beta-1}]$}.
Then, there exists a $t_0>0$ such that for all $t\geq t_0$ it holds
\begin{equation}
\Pb(\{\omega:D_\gamma\in \pi^{\rm max}_{\mathcal{L}^{-}\to E}(\omega)\}) \leq \tilde{C}\exp(-c\, t^{\beta-1/3}),
\end{equation}
for some $t$-independent constants $\tilde{C},c>0$.
\end{prop}
\begin{proof}[Proof of Proposition~\ref{PropLminusC}]
This proof is very close to the one of Proposition~\ref{PropLplusC}. This time we choose $\e= \frac{C}{2} t^{\beta-1}$ with $C=\frac{(\alpha-(2-\alpha)\sqrt{\eta})^2}{2\alpha(2-\alpha)}$ and denote by $\tilde I_{D_\gamma}$ the events such that the maximizers from $\mathcal{L}^-$ to $E$ passes by the point $D_\gamma$. Then,
\begin{equation}
\Pb(\tilde I_{D_\gamma})\leq \Pb\bigg(\tilde I_{D_\gamma}\cap \Big(\bigcap_{D_\gamma} \tilde E_{D_\gamma}\Big)\bigg)+\Pb\bigg(\Big(\bigcap_{D_\gamma} \tilde E_{D_\gamma}\Big)^c\bigg).
\end{equation}
Using Corollary~\ref{CorLminus} we can bound the second term as $\exp(-c' t^{\beta-1/3})$. By Propositions~\ref{PropLminusA} and~\ref{PropLminusB} we obtain
\begin{equation}
\begin{aligned}
L_{\mathcal{L}^-\to E}(\omega) &\leq (\tilde \mu_\gamma+\mu_{{\rm pp},\gamma}+\e)t=\tilde \mu t + (\tilde \mu_\gamma+\mu_{{\rm pp},\gamma}+\e-\tilde\mu)t \\
&\leq \tilde \mu t - (\tilde{C} t^{\beta-1/3}) t^{1/3}
\end{aligned}
\end{equation}
for $\omega \in \tilde I_{D_\gamma}$ and at the same time in each of the $\tilde E_{D_\gamma}$'s.
Therefore,
\begin{equation}\label{eq41b}
\Pb\bigg(\tilde I_{D_\gamma}\cap \Big(\bigcap_{D_\gamma} \tilde E_{D_\gamma}\Big)\bigg) \leq
\Pb\left(L_{\mathcal{L}^-\to E}\leq \mu t - (\tilde{C} t^{\beta-1/3}) t^{1/3}\right).
\end{equation}
The following is slightly different from the previous proof. Denote by
\begin{equation}
Z^-= (\kappa t/2,-\kappa t/2),\quad B=(\kappa t,0),
\end{equation}
where $\kappa=\eta-\alpha^2/(2-\alpha)^2$.
Then, since $L_{\mathcal{L}^{-}\to E}\geq L_{Z^- \to B}+ L_{B\to E}$, it follows that
\begin{equation}\label{eq42b}
\begin{aligned}
(\ref{eq41b})& \leq \Pb\left(L_{Z^-\to B}+L_{B\to E}\leq \tilde \mu t - (\tilde{C} t^{\beta-1/3}) t^{1/3}\right)\\
&\leq \Pb\bigg(L_{Z^-\to B}\leq \tilde \mu_1 t - \frac{\tilde{C} t^{\beta-1/3}}{2} t^{1/3}\bigg)+
\Pb\bigg(L_{B\to E}\leq \tilde \mu_2 t - \frac{\tilde{C} t^{\beta-1/3}}{2} t^{1/3}\bigg),
\end{aligned}
\end{equation}
where $\tilde\mu_1=2\kappa/\alpha$ and $\tilde\mu_2=\tilde\mu-\tilde\mu_1=4/(2-\alpha)^2$. We can finally apply the bound of Proposition~\ref{devtwo} to the two point-to-point problems and finish the proof as in Proposition~\ref{PropLplusC}.
\end{proof}

\begin{proof}[Proof of Proposition~\ref{nocrossingMinus}]
The proof is a straightforward consequence of Proposition~\ref{PropLminusC}, since the cardinality of the family of points $\{D_\gamma\}_{\gamma\in [0,1-t^{\beta-1}]}$ is only of order $t$.
\end{proof}

\subsection{Verification of Assumptions~\ref{Assumpt1}--\ref{Assumpt3}}\label{sectApplications}
\begin{proof}[Proof of Corollary~\ref{Cor1}]
Assumption~\ref{Assumpt1} is fulfilled through Propositions~\ref{propHalfFlatConvergence} and~\ref{propHalfFlatConvergence2}. Note that taking $\hat{\sigma}_\eta,\tilde{\sigma}_\eta$ or
$\hat{\sigma}_{\eta_{0}},\tilde{\sigma}_{\eta_{0}}$ yields the same limits.
Let $\tilde{\mu}_\eta=2(\eta/\alpha+1/(2-\alpha))$ and $\mu_\eta=2(1+\eta)$ be the leading order terms  of the two LPP problems for $\eta$.
The shift in $G_2$ comes from the fact that $\frac{(\mu-\tilde{\mu}_\eta)t}{t^{1/3}}=-\frac{2u}{\alpha}$ and $\frac{(\mu-\mu_\eta)t}{t^{1/3}}=-2u$. Assumption~\ref{Assumpt2} is directly satisfied via Propositions~\ref{propJohConvergence} and~\ref{propHalfFlatConvergence} with $E^{+}=(\eta t-t^{\nu},t-t^{\nu})$. Finally, Assumption~\ref{Assumpt3} is precisely the content of Propositions~\ref{nocrossingPlus} and~\ref{nocrossingMinus}.
\end{proof}

\begin{proof}[Proof of Corollary~\ref{Cor2}]
Clearly any maximizing path $\pi_{\mathcal{L}^{+}\to (\eta t,t)}^{\max}$ starts off at $(-\lfloor \beta_0 t +b t^{1/3}\rfloor ,0)$. Let $\tilde{\mu}_\eta=2(\eta/\alpha+1/(2-\alpha))$ and  $\mu_{\mathrm{pp},\eta}=4+2(u+b)t^{-2/3}$ be the leading order terms of the two LPP problems for $\eta$.Then we have $\frac{(4-\tilde{\mu}_\eta)t}{t^{1/3}}=-\frac{2u}{\alpha}$, $\frac{(4-\mu_{\mathrm{pp},\eta})t}{t^{1/3}}=-2(u+b)$.
Assumption~\ref{Assumpt1} is fulfilled through Propositions~\ref{propJohConvergence} and~\ref{propHalfFlatConvergence2}. The requirement $\alpha<1$ comes from the requirement $\eta_0>\alpha^2/(2-\alpha)^2$ from Proposition~\ref{propHalfFlatConvergence2}.
Assumption~\ref{Assumpt2} is directly satisfied via Propositions~\ref{propJohConvergence} and~\ref{propHalfFlatConvergence}. Finally, Assumption~\ref{Assumpt3} is precisely the content of Propositions~\ref{nocrossingPtPt} and~\ref{nocrossingMinus}.
\end{proof}

\begin{proof}[Proof of Corollary~\ref{Cor3}]
Any maximizing path $\pi_{\mathcal{L}^{+}\to (\eta t,t)}^{\max}$ starts off from $(-\lfloor \beta t \rfloor,0)$.
Let $\mu_{\mathrm{pp},\eta}=(1+\sqrt{1+\beta})^{2}+(1+\frac{1}{\sqrt{1+\beta}})ut^{-2/3}$ be the leading order of $L_{\mathcal{L}^{+}\to(\eta t,t)}$, i.e.
$\frac{(\mu-\mu_{\mathrm{pp},\eta})t}{t^{1/3}}=\mbox{$(1+\frac{1}{\sqrt{1+\beta}})u$}$, so
Assumption~\ref{Assumpt1} is fulfilled through Proposition~\ref{propJohConvergence} with $G_1(s)=F_2(s/\sigma-u(1+1/\sqrt{1+\beta})/\sigma)$.
Note now $L_{\mathcal{L}^{-}\to (\eta t, t)}\overset{d}{=}L_{0\to (\eta t, (1+\beta)t)}$, implying that the leading order of this LPP is $\mu_{\mathrm{pp},\gamma}=(1+\sqrt{1+\beta})^{2}+(1+\sqrt{1+\beta})ut^{-2/3}$  so that
$\frac{(\mu-\mu_{\mathrm{pp},\gamma})t}{t^{1/3}}=-u(1+\sqrt{1+\beta})$, which shows $G_2(s)=F_2(s/\sigma-u(1+\sqrt{1+\beta})/\sigma)$.
Assumption~\ref{Assumpt2} is directly satisfied via Proposition~\ref{propJohConvergence}. Finally, Assumption~\ref{Assumpt3} holds by Proposition~\ref{nocrossingsymmetric}.

\end{proof}

\section{Derivation of the kernel for TASEP with $\alpha$-particles}\label{sectDistrLminus}
In order to prove Proposition~\ref{PropDistrLminus} we first study the system with only $M$ $\alpha-$particles. We denote by $\Pb^{(M)}$ the probability measure for this system. The system we are considering is then recovered by taking the $M\to\infty$. We first recall the generic theorem for joint distributions in TASEP, specialized to our jump rates and initial configuration.

\begin{prop}[Proposition 4 in~\cite{BFS09}]\label{PropFiniteM}
Let us consider particles starting from
\begin{equation}\label{IC}
x_j (0)= 2(M-j), 1 \leq j \leq M,  \quad x_j(0)= -j+M, j> M
\end{equation} and having jump rates
$v_j$  given by
\begin{equation}\label{jumprates}
v_j=\alpha,1\leq j\leq M,\quad v_j=1, j> M.
\end{equation}
Denote $x_{j}(t)$ the position of particle $j$ at time $t$. Then
\begin{equation}
\Pb^{(M)}(x_{n}(t) > s) =\det(\Id -\chi_{s} K_{n,t} \chi_{s})_{\ell^{2}(\mathbb{Z})},
\end{equation}
where $\chi_{s}=\Id(x <s)$.
The kernel $K_{n,t}$ is given by
\begin{equation}\label{newkern}
K_{n,t}(x_1,x_2)= \sum_{k=1}^{n}\Psi_{n-k}^{n,t}(x_1)\Phi_{n-k}^{n,t}(x_2).
\end{equation}
 The functions $\Psi_{n-j}^{n,t}$ are given by
\begin{equation}\label{Psi}
\Psi_{n-j}^{n,t}(x)=\frac{1}{2\pi \I}\oint_{\Gamma_{0}}\frac{\mathrm{d}w}{w}\frac{e^{tw}}{w^{x-x_{j}(0)+n-j}}\prod_{k=j+1}^{n}(w-v_k).
\end{equation}
The functions $\{\Phi_{n-j}^{n,t}\}_{1\leq j \leq n}$ are characterized by the two conditions:
\begin{equation}\label{orthogonality}
 \langle \Psi_{n-j}^{n,t},\Phi_{n-k}^{n,t}\rangle:=\sum_{x\in \mathbb{Z}}\Psi_{n-j}^{n,t}(x)\Phi_{n-k}^{n,t}(x)=\delta_{j,k},\quad 1\leq j,k\leq n,
\end{equation}
and
\begin{equation}\label{spanning}
\mathrm{span}\{ \Phi_{n-j}^{n,t}(x), 1\leq j\leq n \}=\mathrm{span}\{1,x,\ldots,x^{n-M-1},\alpha^{x},x\alpha^{x},\ldots,x^{M-1}\alpha^{x}\}.
\end{equation}
\end{prop}
The following lemma gives explicit formulas for the orthogonal functions $\Phi,\Psi$ defined in the preceeding proposition.  We only give them for $n\geq M+1$, since these are the ones we need.
\begin{lem}
Let $n\geq M+1$. We then have two cases:
\begin{itemize}
\item[(a)] for $j=M+1,\ldots,n,$
\begin{equation}
\begin{aligned}
&\Psi_{n-j}^{n,t}(x)=\frac{1}{2\pi \I}\oint_{\Gamma_{-1}}\frac{\dx w}{w+1}\frac{e^{t(w+1)}}{(w+1)^{x-M+n}}
w^{n-j}
\\&\Phi_{n-j}^{n,t}(x)=\frac{1}{2 \pi \I}\oint_{\Gamma_{0}}\dx z \frac{(z+1)^{x-M+n}}{e^{t(z+1)}z^{n-j+1}}
\end{aligned}
\end{equation}
\item[(b)] for $j=1,\ldots,M,$
\begin{equation}
\begin{aligned}
&\Psi_{n-j}^{n,t}(x)=\frac{1}{2\pi \I}\oint_{\Gamma_{-1}}\frac{\dx w}{w+1}\frac{w^{n-M}(w+1-\alpha)^{M-j}}{(w+1)^{x-2M+n+j}}e^{t(w+1)}
\\&\Phi_{n-j}^{n,t}(x)=\frac{1}{(2 \pi \I)^{2}}\oint_{\Gamma_{\alpha-1}}\dx v
\oint_{\Gamma_{0,v}}\dx z\frac{(z+1)^{x-M+n}}{e^{t(z+1)}z^{n-M}}\\
&\hspace{10em}\times\frac{2v+2-\alpha}{((v+1)(v+1-\alpha))^{M-j+1}}\frac{1}{z-v}
\end{aligned}
\end{equation}
\end{itemize}
\end{lem}
\begin{proof}
The formulas for $\Psi_{n-j}^{n,t}$ are easily obtained by plugging \eqref{IC},\eqref{jumprates} into \eqref{Psi}.

In case $(a)$, using the derivative formula for the residue, one sees that
$\Phi_{n-j}^{n,t}$ is a polynomial of degree $n-j$ and thus
\begin{equation}
\text{span}\{\Phi_{n-j}^{n,t}(x),j=M+1,\ldots,n\}=\text{span}\{1,x,\ldots,x^{n-M-1}\}.
\end{equation}
In case $(b)$, taking the residue at $z=v$, one gets
\begin{align}\label{A}
&\Phi_{n-j}^{n,t}(x)=\frac{1}{2 \pi \I}\oint_{\Gamma_{\alpha-1}}\text{d}v
\frac{(2v+2-\alpha)(v+1)^{x-2M+j-1}}{e^{t(v+1)}v^{n-M}(v+1-\alpha)^{M-j+1}}
\\&\label{B}+\frac{1}{(2\pi \I)^{2}}\oint_{\Gamma_{\alpha-1}}\text{d}v\oint_{\Gamma_{0}}
\text{d}z\frac{(z+1)^{x-M+n}}{e^{t(z+1)}z^{n-M}}\frac{2v+2-\alpha}{((v+1)
(v+1-\alpha))^{M-j+1}}\frac{1}{z-v}.
\end{align}
Now, $\eqref{A}=\alpha^{x}p_{M-j}(x),$ where $p_{M-j}$ is a polynomial of degree
$M-j$. For \eqref{B}, we choose the integration paths such that $|v|>|z|$, apply the identity $(z-v)^{-1}=-v^{-1}\sum_{\ell\geq 0}
(z/v)^\ell$, and obtain
\begin{equation}
\eqref{B}=\sum_{\ell\geq 0}\frac{-1}{(2\pi \I)^{2}} \oint_{\Gamma_{\alpha-1}}\text{d}v\frac{(2v+2-\alpha)v^{-(\ell+1)}}{((v+1)(v+1-\alpha))^{M-\ell+1}}\oint_{\Gamma_{0}}\text{d}z\frac{(z+1)^{x-M+n}}{e^{t(z+1)}
z^{n-M-\ell}},
\end{equation}
which for $\ell=0,\ldots,n-M-1$ is a polynomial of degree $n-M-1-\ell$, and is $0$ for larger $\ell$. Therefore \eqref{spanning} holds.

Next we check the biorthogonality relations \eqref{orthogonality}.
We shall recurrently use
\begin{equation}\label{eq58}
\sum_{x\geq M-n}\left(\frac{z+1}{w+1}\right)^{x-M+n} =  \frac{w+1}{w-z},
\end{equation}
which holds if $|w+1|>|z+1|$.

\textit{Case} $M+1\leq j,k\leq n$:
\begin{equation}\label{eq59}
\begin{aligned}
\langle \Psi_{n-j}^{n,t},\Phi_{n-k}^{n,t}\rangle
&=\sum_{x \in \mathbb{Z}}\frac{1}{(2\pi \I)^{2}}
\oint_{\Gamma_{-1}}\frac{\textrm{d}w}{w+1}\frac{e^{t(w+1)}w^{n-j}}{(w+1)^{x-M+n}}
\oint_{\Gamma_{0}}\mbox{d} z \frac{(z+1)^{x-M+n}}{e^{t(z+1)}z^{n-j+1}}
\\& =\sum_{x \geq M-n}\frac{1}{(2\pi \I)^{2}}
\oint_{\Gamma_{-1}}\frac{\textrm{d}w}{w+1}\frac{e^{t(w+1)}w^{n-j}}{(w+1)^{x-M+n}}
\oint_{\Gamma_{0}}\mbox{d} z \frac{(z+1)^{x-M+n}}{e^{t(z+1)}z^{n-j+1}}
\end{aligned}
\end{equation}
since for $x<M-n$ the functions $\Psi^{n,t}_{n-j}(x)=0$. We can now choose the integration paths such that $|w+1|>|z+1|$. Applying (\ref{eq58}), the pole at $w=-1$ disappears and instead there is a simple pole at $w=z$,
\begin{equation}\label{eq60}
\begin{aligned}
(\ref{eq59})&=\frac{1}{(2\pi \I)^{2}}\oint_{\Gamma_{0}}\mbox{d} z \frac{1}{e^{t(z+1)}z^{n-j+1}}
\oint_{\Gamma_{z}}dw\frac{e^{t(w+1)}w^{n-j}}{w-z}
\\&=\frac{1}{2\pi \I}\oint_{\Gamma_{0}}\text{d}z\frac{1}{z^{j-k+1}}=\delta_{j,k}.
\end{aligned}
\end{equation}

\textit{Case} $M+1\leq j\leq n$ and $1\leq k\leq M$: Also in this case we first restrict the sum over $x\geq M-n$, use (\ref{eq58}), and integrate out the remaining simple pole at $w=z$, with the result
\begin{equation}\label{eq61}
\begin{aligned}
\langle \Psi_{n-j}^{n,t}&,\Phi_{n-k}^{n,t}\rangle
=\sum_{x \in \mathbb{Z}}\frac{1}{(2\pi \I)^{3}}
\oint_{\Gamma_{-1}}\frac{\text{d}w}{w+1}\frac{e^{t(w+1)}w^{n-j}}{(w+1)^{x-M+n}}\\&\times\oint_{\Gamma_{\alpha-1}}\text{d}v
\oint_{\Gamma_{0,v}}\text{d}z\frac{(z+1)^{x-M+n}}{e^{t(z+1)}z^{n-M}}\frac{2v+2-\alpha}{((v+1)(v+1-\alpha))^{M-k+1}}\frac{1}{z-v}
\\&=\frac{1}{(2\pi \I)^{2}}\oint_{\Gamma_{\alpha-1}}\text{d}v\oint_{\Gamma_{0,v}}\text{d}z
\frac{1}{z^{j-M}}\frac{(2v+2-\alpha)}{(z-v)((v+1)(v+1-\alpha))^{M-k+1}}.
\end{aligned}
\end{equation}
Since $j>M$, for $|z|\to\infty$, the integrand in $z$ goes to zero at least as fast as  $1/z^2$ and it does not contain any other poles than $z=0,v$. Therefore, the integrand in $z$ has no pole at infinity and consequently $(\ref{eq61})=0$.

\textit{Case} $1\leq j,k\leq M$: Also in this case we first restrict the sum over $x\geq M-n$, use (\ref{eq58}), and integrate out the remaining simple pole at $w=z$. This gives
\begin{equation}\label{eq62}
\langle \Psi_{n-j}^{n,t},\Phi_{n-k}^{n,t}\rangle
=\frac{1}{(2\pi \I)^{2}}\oint_{\Gamma_{\alpha-1}}\text{d}v
\oint_{\Gamma_{0,v}}\text{d}z \frac{(2v+2-\alpha)((z+1)(z+1-\alpha))^{M-j}}{((v+1)(v+1-\alpha))^{M-k+1}(z-v)}.
\end{equation}
Now, the pole at $z=0$ disappeared and the only contribution comes from the simple pole $z=v$, i.e.,
\begin{equation}
(\ref{eq62})=\frac{1}{2\pi \I}\oint_{\Gamma_{\alpha-1}}\text{d}v\frac{2v+2-\alpha}{((v+1)(v+1-\alpha))^{j-k+1}}
=\frac{1}{2\pi \I}\oint_{\Gamma_0}\text{d}u\frac{1}{u^{j-k+1}}=\delta_{j,k},
\end{equation}
where we used the change of variables $u=(v+1)(v+1-\alpha)$.

\textit{Case} $1\leq j\leq M$ and $M+1\leq k\leq n$: Doing the first steps as in the three other cases above, we get
\begin{equation}
\begin{aligned}
\langle \Psi_{n-j}^{n,t},\Phi_{n-k}^{n,t}\rangle
&=\sum_{x \in \mathbb{Z}}\frac{1}{(2\pi \I)^{2}}
\oint_{\Gamma_{-1}}\frac{\textrm{d}w}{w+1}\frac{w^{n-M}(w+1-\alpha)^{M-j}}{(w+1)^{x+n-2M+j}}e^{t(w+1)}
\\&\hspace{5em}\times \oint_{\Gamma_{0}}\mbox{d} z \frac{(z+1)^{x-M+n}}{e^{t(z+1)}z^{n-k+1}}
\\&=\frac{1}{2 \pi \I}\oint_{\Gamma_0}\text{d}z\frac{((z+1)(z+1-\alpha))^{M-j}}{z^{M-k+1}}=0,
\end{aligned}
\end{equation}
since, for $k>M$ the pole at $z=0$ disappears.
\end{proof}

Later, we will take the $M\to \infty $ limit with $n-M$ finite. To this end we give a compact form of $K_{n,t}$.
\begin{cor}\label{compact}
Let $K_{n,t}$ be the kernel defined in \eqref{newkern}. Then
\begin{equation}
K_{n+M,t}=K_{n,M,t}^{(0)}+K_{n,t}^{(1)}+K_{n,t}^{(2)},
\end{equation}
where $K_{n,t}^{(1)}$ and $K_{n,t}^{(2)}$ are given in (\ref{kernelhalfflatLminus}) and
\begin{equation}\label{eqKernelK0}
\begin{aligned}
K_{n,M,t}^{(0)}(x_1,x_2)&=\frac{-1}{(2\pi \I)^{3}}\oint_{\Gamma_{-1}}\frac{\mathrm{d}w}{w+1}\oint_{\Gamma_{\alpha-1}}
  \mathrm{d}v\oint_{\Gamma_{0,v}}\mathrm{d}z\frac{e^{t(w+1)}w^{n}}{(w+1)^{x_1+n}}\frac{(z+1)^{x_2+n}}{e^{t(z+1)}z^{n}}
  \\&\times\frac{1}{z-v}\frac{2v+2-\alpha}{(v-w)(v+w+2-\alpha)}\bigg(\frac{(w+1)(w+1-\alpha)}{(v+1)(v+1-\alpha)}\bigg)^{M}.
 \end{aligned}
 \end{equation}
\end{cor}
\begin{proof}
We first show that
\begin{equation}K_{n,M,t}^{(0)}(x_1,x_2)+K_{n,t}^{(1)}(x_1,x_2)=\sum_{k=1}^{M}\Psi_{n+M-k}^{n+M,t}(x_1)\Phi_{n+M-k}^{n+M,t}(x_2).\end{equation}
We have
\begin{equation}
\begin{aligned}\label{mykern}
\sum_{k=1}^{M}&\Psi_{n+M-k}^{n+M,t}(x_1)\Phi_{n+M-k}^{n+M,t}(x_2)\\
&=\sum_{k=1}^{M}
\frac{1}{(2\pi \I)^{3}}\oint_{\Gamma_{\alpha-1}}\dx v \oint_{\Gamma_{0,v}}
\dx z\oint_{\Gamma_{-1}}\frac{\dx w}{w+1}
\frac{e^{t(w+1)}w^{n}}{(w+1)^{x_1+n}}\frac{(z+1)^{x_2+n}}{e^{t(z+1)}z^{n}}
\\&\hspace{4em}\times \frac{2v+2-\alpha}{(v+1)(v+1-\alpha)}
\bigg(\frac{(w+1-\alpha)(w+1)}{
(v+1)(v+1-\alpha)}\bigg)^{M-k}\frac{1}{z-v}.
\end{aligned}
\end{equation}
We apply a finite geometric sum formula to $q=\frac{(w+1-\alpha)(w+1)}{
(v+1)(v+1-\alpha)}$. For this the contours need to satisfy $q\neq 1$.
We take the contours such that
\begin{equation}\label{conditions}-\Gamma_{-1}-2+\alpha \subset \Gamma_{\alpha-1}, \Gamma_{-1}\not\subset \Gamma_{\alpha-1},\Gamma_{\alpha-1} \subset \Gamma_{0,v}
,\,\quad \text{and}\quad q\neq 1.\end{equation}
Note that none of these conditions alter \eqref{mykern}.  An explicit choice of paths satisfying \eqref{conditions} is later given in \eqref{paths}.
Using the linearity of the integral, we get
\begin{equation}
\begin{aligned}\label{bravenewworld}
\eqref{mykern}=&\frac{1}{(2\pi \I)^{3}}
\oint_{\Gamma_{-1}}\frac{\dx w}{w+1}
\oint_{\Gamma_{\alpha-1,-w-2+\alpha}}\dx v \oint_{\Gamma_{0,v}}
\dx z
\frac{e^{t(w+1)}w^{n}}{(w+1)^{x_1+n}}\frac{(z+1)^{x_2+n}}{e^{t(z+1)}z^{n}}
\\&\times \frac{2v+2-\alpha}{(v-w)(v+w+2-\alpha)}
\bigg(1-\bigg(\frac{(w+1-\alpha)(w+1)}{(v+1)(v+1-\alpha)}\bigg)^{M}\bigg)\frac{1}{z-v}
\\=&\frac{1}{(2\pi \I)^{3}}
\oint_{\Gamma_{-1}}\frac{\dx w}{w+1}
\oint_{\Gamma_{-w-2+\alpha}}\dx v \oint_{\Gamma_{0,v}}
\dx z
\frac{e^{t(w+1)}w^{n}}{(w+1)^{x_1+n}}\frac{(z+1)^{x_2+n}}{e^{t(z+1)}z^{n}}
\\&\times \frac{2v+2-\alpha}{(v-w)(v+w+2-\alpha)}\frac{1}{z-v}
\\&-\frac{1}{(2\pi \I)^{3}}
\oint_{\Gamma_{-1}}\frac{\dx w}{w+1}
\oint_{\Gamma_{\alpha-1}}\dx v \oint_{\Gamma_{0,v}}
\dx z
\frac{e^{t(w+1)}w^{n}}{(w+1)^{x_1+n}}\frac{(z+1)^{x_2+n}}{e^{t(z+1)}z^{n}}
\\&\times \frac{2v+2-\alpha}{(v-w)(v+w+2-\alpha)}
\bigg(\frac{(w+1-\alpha)(w+1)}{(v+1)(v+1-\alpha)}\bigg)^{M}\frac{1}{z-v}.
\end{aligned}
\end{equation}
Here we used that in the first triple integral the pole $v=\alpha-1$ is no longer present. Plugging in the remaining residue at $v=-w-2+\alpha$ yields then
\begin{equation}
\eqref{bravenewworld}=K_{n,M,t}^{(0)}(x_1,x_2)+K_{n,t}^{(1)}(x_1,x_2).
\end{equation}

Next we define
\begin{equation}
K_{n,t}^{(2)}(x_1,x_2):=\sum_{k=M+1}^{n+M}\Psi_{n+M-k}^{n+M-k,t}(x_1)\Phi_{n+M-k}^{n+M-k,t}(x_2).
\end{equation}
Note that
 $\Phi_{n+M-k}$ is zero for $k\geq n+M+1$, thus
\begin{equation}
K_{n,t}^{(2)}(x_1,x_2)=\sum_{k=M+1}^{\infty}\frac{1}{(2\pi \I)^{2}}\oint_{\Gamma_{0}}\mathrm{d}z
 \oint_{\Gamma_{-1}}\frac{\mathrm{d}w}{w+1}\frac{e^{t(w+1)}w^{n+M-k}}{(w+1)^{x_1+n}}\frac{(z+1)^{x_2+n}}{e^{t(z+1)}z^{n+M-k+1}}.
 \end{equation}
Assuming the contours are such that $|w|>|z|,$ taking geometric series yields
\begin{equation}
K_{n,t}^{(2)}(x_1,x_2)=\frac{1}{(2\pi \I)^{2}}\oint_{\Gamma_{0}}\mathrm{d}z
 \oint_{\Gamma_{-1,z}}\frac{\mathrm{d}w}{w+1}\frac{e^{t(w+1)}w^{n}}{(w+1)^{x_1+n}}\frac{(z+1)^{x_2+n}}{e^{t(z+1)}z^{n}}\frac{1}{w-z}.
\end{equation}
Finally, it is straightforward to check that the contribution of the simple pole at $w=z$ is zero, so that we can drop it in the final expression of $K_{n,t}^{(2)}$.
\end{proof}

\begin{prop}
Let $K_{n,M,t}^{(0)},K_{n,t}^{(1)}, K_{n,t}^{(2)}$ be as in \eqref{kernelhalfflatLminus} and \eqref{eqKernelK0}. Then, for $x_1,x_2 \leq\ell,$ we have the following bounds.
\begin{equation}
\begin{aligned}\label{kernsbounds}
&|K_{n,M,t}^{(0)}(x_1,x_2)|\leq C\,e^{c x_2}q^{M}
\\&|K_{n,t}^{(1)}(x_1,x_2)|\leq C\, e^{c x_2}
\\&|K_{n,t}^{(2)}(x_1,x_2)|\leq C\, e^{c x_2}
\end{aligned},
\end{equation} with  $q \in [0,1)$,
 $c>0$ a constant,   and $C$ depends only on $\ell,n,t$.
\end{prop}
\begin{proof}
To bound $|K_{n,M,t}^{(0)}(x_1,x_2)|,$ we set
\begin{equation}\label{paths}
\Gamma_{-1}=-1+r_{1}e^{\I s_1}\quad\Gamma_{\alpha -1}=\alpha -1+r_{2}e^{\I s_2}\quad\Gamma_{0,v}=r_{3}e^{\I s_3}
\end{equation}
with  $r_1=\frac{\alpha^{2}}{10},r_2 =\frac{\alpha}{\sqrt{1.5}}, r_3 =1-\alpha+r_2+\frac{|r_1+r_2-\alpha|}{2}.$ It is straightforward to check that \eqref{paths} satisfy \eqref{conditions}.
We will bound the different parts of $K_{n,M,t}^{(0)}$. First we note
\begin{equation}
\begin{aligned}\label{term2}
&q:=\frac{\max_{\Gamma_{-1}}|(w+1)(w+1-\alpha)|}{\min_{\Gamma_{-1}}|(v+1)(v+1-\alpha)|}<\frac{\sqrt{1.5}|\alpha(-1-\alpha/10)|}{10(1-1/\sqrt{1.5}))}<1
\\&\frac{\max_{\Gamma_{0,v}}|(z+1)^{x_2+n}|}{\min_{\Gamma_{0,v}}|e^{t(z+1)}z^{n}|}\leq C \,(1+r_{3})^{x_2}\leq C e^{c x_2}
\end{aligned}
\end{equation}
The remaining parts can now  be bounded by a constant:
\begin{equation}
\begin{aligned}\label{term1}
&\frac{\max_{\Gamma_{\alpha -1}}|2v+2-\alpha|}
{\min_{\Gamma_{-1},\Gamma_{\alpha-1}}|(v-w)(v+w+2-\alpha)|}\leq \frac{\alpha+2r_2}{(r_2-r_1)(\alpha-\alpha/\sqrt{1.5}-\alpha^{2}/10)} < C
\\&\frac{1}{\min_{\Gamma_{-1},\Gamma_{0,v}}|z-v|} < C,\\
&\frac{\max_{\Gamma_{-1}}|e^{t(w+1)}w^{n}|}{\min_{\Gamma_{-1}}|(w+1)^{x_1+n}|}\leq \tilde{C}\,
r_{1}^{-x_1}\leq C,
\end{aligned}
\end{equation}
where the last estimate in \eqref{term1} holds
since $0<r_1<1$ and $x_1 \leq \ell$.
Putting these bounds together gives the estimate for $K_{n,M,t}^{(0)}$. Note that the contour for $z$ contains $\alpha-2-w$. Therefore, in $K_{n,t}^{(1)}$ we can choose the same contours for $z,w$ as before and use the estimates from \eqref{term2}, \eqref{term1}. Noting
\begin{equation}\min_{\Gamma_{-1},\Gamma_{0,\alpha-2-w}}|z-(\alpha-2-w)|^{-1}\leq C,\end{equation} one gets the same bound as for $K_{n,M,t}^{(0)}$, only without the $q^{M}$.

As for $K_{n,t}^{(2)}$, we can again choose the same contours for $z,w$ as before. Since $|w-z|$ is bounded from below, we get the same estimate as for $K_{n,t}^{(1)}$.
\end{proof}

Now we are ready to proof Proposition~\ref{PropDistrLminus}.
\begin{proof}[Proof of Proposition~\ref{PropDistrLminus}]
Denote for clarity by $x_{n+M}^{M}(t)$ the position of particle number $n+M$ at time $t$  in the system with $M$ slow particles (defined via \eqref{IC} and \eqref{jumprates},) and by $x_{n}(t)$ the position of particle $n$ at time $t$ in the system with infinitely many slow particles (defined via \eqref{eq43} and \eqref{eq44}).
First we note that
\begin{equation}
\lim_{M\to\infty}\Pb^{(M)}(x_{n+M}^{M}(t) > s)=\Pb(x_{n}(t) > s).
\end{equation}
This follows since $x_{n+M}^{M}(0)=x_{n}(0)$ and by the fact that in TASEP the position of a particle up to a fixed time $t$ depends only on finitely many other particles with probability one, as is seen from a graphical construction of it.
Therefore, by Corollary \ref{compact}, it remains to prove
\begin{equation}
\lim_{M\to\infty} \det(\Id-\chi_s K_{n+M,t} \chi_s)_{\ell^2(\Z)}=\det(\Id-\chi_s \tilde K_{n,t} \chi_s)_{\ell^2(\Z)},
\end{equation}
where we used the notation $K_{n+M,t}=K_{n,M,t}^{(0)}+K_{n,t}^{(1)}+ K_{n,t}^{(2)}$.

By the bounds in \eqref{kernsbounds}, we know that $K_{n,M,t}^{(0)}$ converges pointwise to $0$. Thus, it remains to show that also the Fredholm determinant converges. Consider the Fredholm series expansion
\begin{equation}\label{eq79}
\det(\Id-\chi_s K_{n+M,t} \chi_s)_{\ell^2(\Z)} = \sum_{m\geq 0}\frac{(-1)^m}{m!} \sum_{x_1\leq s}\ldots \sum_{x_m\leq s} \det[K_{n+M,t}(x_i,x_j)]_{1\leq i,j\leq m}.
\end{equation}
By \eqref{kernsbounds}, we have
\begin{equation}
\left|\frac{(-1)^{n}}{n!}\det\big(K_{n+M,t}(x_k,x_l)\big)_{1\leq k,l\leq n}\right|\leq \frac{1}{n!}e^{c(x_1+\cdots+x_n)}C^{n}(2+q^{M})^{n}n^{n/2},
\end{equation} where $n^{n/2}$ is the Hadamard bound for matrices with entries of absolute value less or equal  than $1$. Since $q<1$, we may replace $2+q^{M}$ by $3$ to get a summable uniform bound. Thus we may apply dominated convergence to (\ref{eq79}) to take the $M\to\infty$ inside the sum, which proves the result.
\end{proof}

%\bibliographystyle{patplain}
%\bibliography{Biblio}

\begin{thebibliography}{10}

\bibitem{AV87}
E.D. Andjel and M.E. Vares, \emph{Hydrodynamic equations for attractive
  particle systems on $\mathbb{Z}$}, J. Stat. Phys. \textbf{47} (1987),
  265--288.

\bibitem{BBP06}
J.~Baik, G.~{Ben Arous}, and S.~P\'ech\'e, \emph{Phase transition of the
  largest eigenvalue for non-null complex sample covariance matrices}, Ann.
  Probab. \textbf{33} (2006), 1643--1697.

\bibitem{BFP09}
J.~Baik, P.L. Ferrari, and S.~P{\'e}ch{\'e}, \emph{{Limit process of stationary
  TASEP near the characteristic line}}, Comm. Pure Appl. Math. \textbf{63}
  (2010), 1017--1070.

\bibitem{BFP12}
J.~Baik, P.L. Ferrari, and S.~P{\'e}ch{\'e}, \emph{{Convergence of the
  two-point function of the stationary TASEP}}, arXiv:1209.0116 (2012).

\bibitem{BR00}
J.~Baik and E.M. Rains, \emph{Limiting distributions for a polynuclear growth
  model with external sources}, J. Stat. Phys. \textbf{100} (2000), 523--542.

\bibitem{BR99}
J.~Baik and E.M. Rains, \emph{Symmetrized random permutations}, Random Matrix
  Models and Their Applications, vol.~40, Cambridge University Press, 2001,
  pp.~1--19.

\bibitem{vanB91}
{H. van} Beijeren, \emph{Fluctuations in the motions of mass and of patterns in
  one-dimensional driven diffusive systems}, J. Stat. Phys. \textbf{63} (1991),
  47--58.

\bibitem{BS13}
V.~Belitsky and G.M. Sch{\"u}tz, \emph{{Microscopic structure of shocks and
  antishocks in the ASEP conditioned on low current}}, J. Stat. Phys.
  \textbf{152} (2013), 93--111.

\bibitem{BC09}
G.~{Ben Arous} and I.~Corwin, \emph{{Current fluctuations for TASEP: a proof of
  the Pr\"ahofer-Spohn conjecture}}, Ann. Probab. \textbf{39} (2011), 104--138.

\bibitem{BF07}
A.~Borodin and P.L. Ferrari, \emph{{Large time asymptotics of growth models on
  space-like paths I: PushASEP}}, Electron. J. Probab. \textbf{13} (2008),
  1380--1418.

\bibitem{BFP06}
A.~Borodin, P.L. Ferrari, and M.~Pr{\"a}hofer, \emph{{Fluctuations in the
  discrete TASEP with periodic initial configurations and the Airy$_1$
  process}}, Int. Math. Res. Papers \textbf{2007} (2007), rpm002.

\bibitem{BFPS06}
A.~Borodin, P.L. Ferrari, M.~Pr{\"a}hofer, and T.~Sasamoto, \emph{{Fluctuation
  Properties of the TASEP with Periodic Initial Configuration}}, J. Stat. Phys.
  \textbf{129} (2007), 1055--1080.

\bibitem{BFS07}
A.~Borodin, P.L. Ferrari, and T.~Sasamoto, \emph{{Transition between Airy$_1$
  and Airy$_2$ processes and TASEP fluctuations}}, Comm. Pure Appl. Math.
  \textbf{61} (2008), 1603--1629.

\bibitem{BFS09}
A.~Borodin, P.L. Ferrari, and T.~Sasamoto, \emph{{Two speed TASEP}}, J. Stat.
  Phys. \textbf{137} (2009), 936--977.

\bibitem{cor11}
I.~Corwin, \emph{{The Kardar-Parisi-Zhang equation and universality class}},
  arXiv:1106.1596 (2011).

\bibitem{CFP10b}
I.~Corwin, P.L. Ferrari, and S.~P{\'e}ch{\'e}, \emph{{Universality of slow
  decorrelation in KPZ models}}, Ann. Inst. H. Poincar\'e Probab. Statist.
  \textbf{48} (2012), 134--150.

\bibitem{DJLS93}
B.~Derrida, S.A. Janowsky, J.L. Lebowitz, and E.R. Speer, \emph{Exact solution
  of the totally asymmetric simple exclusion process: shock profiles}, J. Stat.
  Phys. \textbf{73} (1993), 813--842.

\bibitem{Fer86}
P.A. Ferrari, \emph{The simple exclusion process as seen from a tagged
  particle}, Ann. Probab. \textbf{14} (1986), 1277--1290.

\bibitem{Fer90}
P.A. Ferrari, \emph{Shock fluctuations in asymmetric simple exclusion}, Probab.
  Theory Relat. Fields \textbf{91} (1992), 81--101.

\bibitem{FF94b}
P.A. Ferrari and L.~Fontes, \emph{{Shock fluctuations in the asymmetric simple
  exclusion process}}, Probab. Theory Relat. Fields \textbf{99} (1994),
  305--319.

\bibitem{FKS91}
P.A. Ferrari, C.~Kipnis, and E.~Saada, \emph{Microscopic structure of
  travelling waves in the asymmetric simple exclusion process}, Ann. Probab.
  \textbf{19} (1991), 226--244.

\bibitem{Fer08}
P.L. Ferrari, \emph{{Slow decorrelations in KPZ growth}}, J. Stat. Mech.
  (2008), P07022.

\bibitem{Fer07}
P.L. Ferrari, \emph{{The universal Airy$_1$ and Airy$_2$ processes in the
  Totally Asymmetric Simple Exclusion Process}}, Integrable Systems and Random
  Matrices: In Honor of Percy Deift (J.~Baik, T.~Kriecherbauer, L-C. Li,
  K.~McLaughlin, and C.~Tomei, eds.), Contemporary Math., Amer. Math. Soc.,
  2008, pp.~321--332.

\bibitem{FS05b}
P.L. Ferrari and H.~Spohn, \emph{{A determinantal formula for the GOE
  Tracy-Widom distribution}}, J. Phys. A \textbf{38} (2005), L557--L561.

\bibitem{FS05a}
P.L. Ferrari and H.~Spohn, \emph{Scaling limit for the space-time covariance of
  the stationary totally asymmetric simple exclusion process}, Comm. Math.
  Phys. \textbf{265} (2006), 1--44.

\bibitem{PG90}
J.~G{\"a}rtner and E.~Presutti, \emph{Shock fluctuations in a particle system},
  Ann. Inst. H. Poincar{\'e} (A) \textbf{53} (1990), 1--14.

\bibitem{Jo00b}
K.~Johansson, \emph{Shape fluctuations and random matrices}, Comm. Math. Phys.
  \textbf{209} (2000), 437--476.

\bibitem{Jo00}
K.~Johansson, \emph{Transversal fluctuations for increasing subsequences on the
  plane}, Probab. Theory Related Fields \textbf{116} (2000), 445--456.

\bibitem{Jo03b}
K.~Johansson, \emph{Discrete polynuclear growth and determinantal processes},
  Comm. Math. Phys. \textbf{242} (2003), 277--329.

\bibitem{KPZ86}
M.~Kardar, G.~Parisi, and Y.Z. Zhang, \emph{Dynamic scaling of growing
  interfaces}, Phys. Rev. Lett. \textbf{56} (1986), 889--892.

\bibitem{Lig76}
T.M. Liggett, \emph{Coupling the simple exclusion process}, Ann. Probab.
  \textbf{4} (1976), 339--356.

\bibitem{Li99}
T.M. Liggett, \emph{Stochastic interacting systems: contact, voter and
  exclusion processes}, Springer Verlag, Berlin, 1999.

\bibitem{PS01}
M.~Pr{\"a}hofer and H.~Spohn, \emph{Current fluctuations for the totally
  asymmetric simple exclusion process}, In and out of equilibrium
  (V.~Sidoravicius, ed.), Progress in Probability, Birkh{\"a}user, 2002.

\bibitem{Sas05}
T.~Sasamoto, \emph{Spatial correlations of the {1D KPZ} surface on a flat
  substrate}, J. Phys. A \textbf{38} (2005), L549--L556.

\bibitem{Spo91}
H.~Spohn, \emph{{Large Scale Dynamics of Interacting Particles}}, Texts and
  Monographs in Physics, Springer Verlag, Heidelberg, 1991.

\bibitem{TW94}
C.A. Tracy and H.~Widom, \emph{{Level-spacing distributions and the Airy
  kernel}}, Comm. Math. Phys. \textbf{159} (1994), 151--174.

\bibitem{TW96}
C.A. Tracy and H.~Widom, \emph{On orthogonal and symplectic matrix ensembles},
  Comm. Math. Phys. \textbf{177} (1996), 727--754.

\end{thebibliography}

\end{document}